\documentclass{article}
\usepackage{sectsty}
\usepackage{amsfonts, amsmath, amsthm, bm, amssymb}
\usepackage{mathtools}
\usepackage{authblk} 
\usepackage{natbib}
\usepackage{hyperref}
\usepackage{xcolor}
\usepackage{psfrag,xspace}

\newtheorem{Corollary}{Corollary}
\newtheorem{Lemma}{Lemma}
\newtheorem{Theorem}{Theorem}
\newtheorem{Proposition}{Proposition}
\newtheorem{Definition}{Definition}

\newcommand{\myP}[2][R]{P_{#1}(#2)}
\newcommand{\myNormP}[2][R]{\tilde{P}_{#1}(#2)}

\newcommand{\distiid}{\overset{\text{i.i.d}}{\sim}}
\newcommand{\convprob}{\overset{p}{\to}}
\newcommand{\convdist}{\overset{w}{\to}}
\newcommand{\Expect}[1]{\mathbb{E}\left[ #1 \right]}
\newcommand{\Var}[1]{\mathrm{Var}\left( #1 \right)}
\newcommand{\Prob}[1]{\mathbb{P}( #1 )}
\newcommand{\iset}{\mathbf{i}}
\newcommand{\jset}{\mathbf{j}}

\DeclareMathOperator*{\argmin}{argmin} 

\begin{document}

\title{Bootstrapping Exchangeable Random Graphs}
\author[1]{Alden Green}
\author[1,2,3]{Cosma Rohilla Shalizi}
\affil[1]{Statistics Department, Carnegie Mellon University, 5000 Forbes Avenue, Pittsburgh, PA 15213 USA}
\affil[2]{Machine Learning Department, Carnegie Mellon University, 5000 Forbes Avenue, Pittsburgh, PA 15213 USA}
\affil[3]{Santa Fe Institute, 1399 Hyde Park Road, Santa Fe, NM 87501 USA}


\maketitle

\begin{abstract}
  We introduce two new bootstraps for exchangeable random graphs.  One, the
  ``empirical graphon bootstrap'', is based purely on resampling, while the other, the
  ``histogram bootstrap'', is a model-based ``sieve'' bootstrap.
  We show that both of them accurately approximate the sampling distributions
  of motif densities, i.e., of the normalized counts of the number of times
  fixed subgraphs appear in the network.  These densities characterize the
  distribution of (infinite) exchangeable networks.  Our bootstraps therefore
  give a valid quantification of uncertainty in inferences
  about fundamental network statistics, and so of parameters identifiable from
  them.
\end{abstract}

\section{Introduction and Goal}

By this point, it is a cliche to say that networks matter, and that network
data analysis is an increasingly important part of statistics. Statistical work
has largely focused on elaborating models and obtaining point-estimates of
their parameters \citep{Olding-Wolfe-inference-for-graphs,
	Kolaczyk-on-network-data}; there has been comparatively little progress in
quantifying uncertainty in these estimates, though that is essential to their
scientific utility.  If we had widely-accepted parametric models, we might hope
to use standard asymptotics, at least heuristically, but we do not have such
models, and we have reason to doubt that standard asymptotics apply to networks
\citep{your-favorite-ergm-sucks}\footnote{Since the standard asymptotics
	essentially rest on the log-likelihood having a quadratic maximum
	\citep{Geyer-on-Le-Cam}, this says something interesting and strange about
	network dependence, but that will have to be pursued elsewhere.}. In other
areas of statistics, bootstrapping has been highly successful at quantifying
uncertainty, even in the face of model mis-specification and complicated
dependence structures \citep{Lahiri-resampling-for-dependent}.  Accordingly, in
this paper, we introduce two bootstraps, one, the ``empirical graphon bootstrap'', based
purely on resampling, the other, the ``histogram bootstrap'',
being a model-based ``sieve'' bootstrap.  We prove that they both accurately
approximate the sampling distributions of ``motif densities'', the normalized
count of the number of times any fixed subgraph (or ``motif'') appears in the
network.  Under exchangeability of the nodes, such densities are known to
characterize the distribution of (infinite) networks, as well as defining the
convergence of sequences of individual (non-random) graphs.  Our bootstraps
therefore provide theoretically sound ways to quantify the
uncertainty in inferences regarding a fundamental class of network statistics,
and so of parameters identifiable from these statistics.

As a contribution to network data analysis, our work extends previous proposals
for quantifying uncertainty by means of subsampling the network and using
plug-in Gaussian approximations \citep{Bhattacharyya-subsample-count-features},
and heuristic parametric bootstraps \citep{Rosvall-Bergstrom-mapping-change}
and resampling schemes \citep{Eldardiry-Neville-network-resampling}.
However, from work on bootstrapping in other areas, we know that estimation of
the distribution via resampling can be more accurate, particularly with small
sample sizes. While our proofs rely on asymptotic arguments via normal
approximations, we think its reasonable to expect that our estimators will
perform well in the small sample setting as well.  From a bootstrap-theory
perspective, our contribution is to extend the validity of resampling and sieve
bootstraps to a new type of dependence structure, joining previous work on time
series, spatial data, and U-statistics.

\paragraph{Probabilistic background and general approach}

Exchangeability of the nodes\footnote{Often called ``joint'' exchangeability,
	to indicate that one applies the same permutation to both the rows and the
	columns of the adjacency matrix of a graph, as opposed to the stricter
	``separate'' exchangeability, where rows and columns can be differently
	permuted \citep[sec.\ 7.1]{Kallenberg-symmetries}.} is a common assumption on
networks; it corresponds to the assumption that any two isomorphic graphs
should be equally probable, and that no information we have on individual nodes
(other than their location in the network) is useful for predicting their
links.  As with other probabilistic symmetries, exchangeability is useful, in
part, because of representation theorems which state that all (infinite)
exchangeable distributions are mixtures of certain extremal distributions with
the same symmetry, but stronger independence properties
\citep{Dynkin-suff-stats-and-extreme-points, Kallenberg-symmetries,
	Lauritzen-extreme-point-models}.  In the case of exchangeable networks, the
relevant extremal distributions, now often called ``graphons'', were
characterized by Aldous and Hoover, and work as follows \citep[ch.\
7]{Kallenberg-symmetries}.  Every node gets an independent, uniformly
distributed random variable on $[0,1]$, say $\epsilon_i$ for node $i$, and
there is a link function $h: [0,1]\times[0,1] \mapsto [0,1]$, symmetric in its
arguments, such that the probability of an edge between $i$ and $j$ is
$h(\epsilon_i, \epsilon_j)$, and all edges are independent (given the
$\epsilon$s).  Any exchangeable distribution is a mixture of such graphon
distributions\footnote{Naturally, the same distribution can be realized by many
	different $h$ functions, which leads to some subtleties in a formal statement
	of the representation theorem.  We do not elaborate on this, since it is not
	relevant to our concerns, but see \citet[sec.\ 7.6]{Kallenberg-symmetries}.},
and any one realization of an exchangeable distribution is drawn from a {\em
	single} $h$.

This provides a natural approach to bootstrapping: estimate the link function
$h$, then randomly redraw node variables and reconnect the edges with the
corresponding probability.  Our task is thus just (!) to estimate the link
function sufficiently well.  We propose two approaches.  One, the ``empirical
graphon'', takes the adjacency matrix, views it as a binary-valued function on
the unit square, and uses that as our estimate of $h$.  Proving the validity
of this bootstrap then relies on results about the convergence, in a suitable
topology, of exchangeable random graphs to their generating graphon.  Our other
bootstrap is a histogram-like estimator of the graphon, a special case of
stochastic block models, essentially approximating $h$ by a series of simple
functions.  Its validity rests on some smoothness assumptions regarding $h$,
but, when they hold, it gives a faster rate of convergence.  In both
approaches, a Berry-Esseen inequality for U-statistics due to
\citet{Callaert-berry-esseen-u-statistics} provides a crucial technical tool.

\paragraph{Organization}

Section \ref{sec:notation} fixes notation, lays out assumptions, and, in \S
\ref{sec:resampling-procedures}, formally proposes the two bootstraps.  Section
\ref{sec:main-results} gives the main theorems, stating conditions under which
our bootstraps consistently approximate the distribution of motif densities. Section~\ref{sec:simulations} shows empirically that our bootstraps perform well even on moderate sized networks.
Section \ref{sec:supporting} collects supporting propositions and lemmas,
and Section \ref{sec:proofs} proves the main results.

\section{Notation and Methodology}
\label{sec:notation}

We (mostly) follow the notation of
\citet{Bhattacharyya-subsample-count-features}.  Unless otherwise noted, by
``graph'' we will always mean an undirected, simple graph.  For any graph $G$,
$V(G)$ will be the set of its vertices, and $E(G)$ the set of its edges; when
$i \in V(G)$, $j \in V(G)$, we write $(i, j)$ for the (unordered) pair of
nodes, and $(i,j) \in E(G)$ or $(i,j) \notin E(G)$ depending on whether or not
there is an edge.  We will sometimes abbreviate this as $(i, j) \in G$ when
there is no chance of ambiguity.  We will use colons to abbreviate sequences,
so that $i:j$ stands for $i, i+1, \ldots j-1, j$, and (say) $x_{i:j}$ the
sequence of variables $x_{i}, x_{i+1}, \ldots x_{j-1}, x_j$.  Given an ordered
$p$-tuple of indices in $1:n$, $\iset = (i_i, i_2, \ldots i_p)$, we let
$G(\iset) = (\iset, E(G) \cap \{\iset \times \iset\})$ be the induced subgraph of $G$ with those vertices; we will write
$S_n(p)$ for the collection of all ordered $p$-tuples of $1:n$.  Two graphs
$G_1$ and $G_2$ are isomorphic when their nodes can be put in one-to-one
correspondence while preserving both edges and non-edges, i.e., there is an
invertible mapping $\sigma: V(G_1) \mapsto V(G_2)$ such that $(i,j) \in E(G_1)$
if and only if $(\sigma(i), \sigma(j)) \in E(G_2)$.  When this holds, we write
$G_1 \simeq G_2$, and we write $N(G)$ for the number of graphs on $1:|V(G)|$
which are isomorphic to $G$.  $K_p$ will indicate the complete graph on $1:p$,
i.e., the $p$-node graph with all possible edges.

Our data $G_n$ is a graph on the vertices $1:n$, with corresponding $n\times n$
adjacency matrix $A$.  We assume that the graph is exchangeable, and
hence was generated as follows:
\begin{eqnarray}
\label{eqn:graphon_sample}
\epsilon_i & \distiid & \mathrm{Uniform}(0,1)\\
A_{ij} | \epsilon_{1:n} & \overset{\mathrm{ind}}{\sim} & \mathrm{Bernoulli}(h_n(\epsilon_i, \epsilon_j))
\end{eqnarray}
for a symmetric, measurable link function $h_n: [0,1]\times[0,1] \mapsto [0,1]$.
Without loss of generality, we decompose the function $h_n$ as
\begin{equation}
h_{n}(u,v) = \rho_{n} w(u,v)
\end{equation}
where $\int_{0}^{1}{\int_{0}^{1}{w(u,v) du dv}} = 1$, so that $\rho_n$ is the
marginal probability of an edge between any two nodes, i.e., the (expected)
edge density\footnote{The job of the $\rho_n$ factor is
	to allow the graph to become sparse as $n$ grows, as in
	\citet{Borgs-Chayes-Zhao-sparse-graph-convergence}; otherwise, graph
	sequences generated by graphons are ``dense'', i.e., the number of edges
	grows quadratically with the number of nodes.  If this is not a concern, and
	this is a point of some debate in the field, one can fix $\rho_n$ to a
	constant value for all $n$.  All of our results are valid under such
	dense-graph limits, and indeed most of them would simplify.}. We will frequently make assume that $w$ is bounded in various $L^p([0,1]^2)$ norms, and we use $\|w\|_{p} = (\int_{[0,1]^2} w(u,v)^p \,du\,dv )^{1/p}$ to refer to such norms.

Fixing any $p$-node connected, simple, undirected graph $R$ that we like, we
can ask about the probability that the first $p$ nodes of $G_n$ instantiate
this motif\footnote{The literature typically calls both $R$ and $G_n(1:p)$
	``subgraphs''; to avoid confusion, we borrow the term ``motif'' from
	\citet{Milo-et-al-motifs} to designate the pattern being matched, though
	those authors suggested using it for the patterns which were, in some sense,
	more common than expected by chance.},
\begin{equation}
\myP{h} = \Prob{G_n(1:p) = R}
\end{equation}
Of course, by exchangeability, $\myP{h} = \Prob{G_n(\iset) = R}$ for any
$\iset \in S_n(p)$.  These probabilities are thus very much like moments of the
distribution of $G_n$, and indeed it is known from previous work
\citep{Lovasz-very-large-graphs} that the collection of these probabilities,
over all motifs $R$, suffice to characterize an exchangeable graph
distribution\footnote{See \citet{Bickel-Chen-Levina-method-of-moments} for a
	discussion, and a method-of-moments procedure for estimating $h$, based on
	this fact.}.  One can show \citep{Lovasz-very-large-graphs} that
\begin{equation}
\myP{h} = \Expect{ \prod_{(i,j) \in E(R)} h_n(\epsilon_i,\epsilon_j) \prod_{(i,j) \notin E(R)} (1-h_n(\epsilon_i,\epsilon_j)) }
\end{equation}

It is natural to want to relate these moments to their sample counter-parts.
It turns out that a good estimate for $\myP{h}$ is simply to count the number
of induced subgraphs in $G_n$ which are isomorphic to
$R$:
\begin{equation}
\myP{G_n} =  \frac{1}{{n \choose p} p! N(R)} \sum_{\iset \in S_n(p)}{\mathbb{I}(G_n(\iset)\simeq R)}
\end{equation}
Unsurprisingly, $\Expect{\myP{G_n}} = \myP{h}$.  Moreover, previous
work\footnote{See, for instance, \citet[Lemma
	4.4]{Borgs-Chayes-Lovasz-et-al-convergent-graph-sequences-i}, which gives an
	explicit (though potentially loose) rate of convergence.  This rate is fast
	enough that a Borel-Cantelli argument could strengthen convergence in
	probability to almost-sure convergence, but this goes beyond what we need
	here.} on graph limits tells us that, for fixed $R$,
$\myP{G_n} \convprob \myP{h}$.

Finally, we will need a few scaled versions of the above quantities, since we
allow the sparsity factor $\rho_n$ to approach $0$ as $n$ grows to $\infty$.  First, let
$\hat{\rho}_n = \myP[K_2]{G_n}$ be the edge density observed in the graph.
Second, let $\myNormP{h} = \frac{\myP{h}}{\rho_{n}^{|E(R)|}}$, and its
corresponding empirical quantity
$\myNormP{G_n} = \frac{\myP{G_n}}{\hat{\rho}_{n}^{|E(R)|}}$.

\paragraph{Miscellaneous notation and conventions} Unless otherwise noted, all
limits are taken as the number of nodes in the graph grows, i.e., as
$n\rightarrow \infty$.

\subsection{Resampling Procedures}
\label{sec:resampling-procedures}

Our resampling procedures begin with an estimate of the graphon, $\hat{h}: [0,1]^2 \rightarrow [0,1]$,
a mapping to be estimated using the graph $G_n$. Given such an estimate $\hat{h}$, we then generate $m$ random
variables $\epsilon_i^* \distiid \mathrm{Uniform}(0,1)$; here $m \in \mathbb{N}$ will be the number of nodes in our resampled network. 
We then simulate from $\hat{h}$ in the way we generate from graphons, forming the bootstrapped
network $G^*_{m}$. More precisely, we let the bootstrapped adjacency matrix
$A^* = \left(A_{ij}^* \right)$ where the entries $A_{ij}^{*}$ are conditionally independent given $G_n$ and $\epsilon_{1:m}^*$, and follow the distribution
$A_{ij}^*|(G_n,\epsilon_{1:m}^*) \sim
\text{Bern}(\hat{h}(\epsilon_i^*,\epsilon_j^*))$. (Implicitly, $m = m(n)$). The
properties of our bootstrapping procedure clearly depend on the graphon estimate $\hat{h}$, and we now formally define the two estimators we will subsequently analyze.

Our first approach estimates the graphon by using its empirical counterpart: the adjacency matrix.

\begin{Definition}[Empirical Graphon]
	The \textbf{empirical graphon}, denoted $\hat{h}_{adj}$, is
	$\hat{h}_{adj}(u,v) := A_{\lceil nu \rceil \lceil nv \rceil}$. We refer to process of resampling using the empirical graphon as the \emph{empirical graphon bootstrap}. 
\end{Definition}

The empirical graphon bootstrap is equivalent to sampling $m$ vertices from $G_n$ (with
replacement), and adding in adjacencies exactly as they appear in $G_n$. Despite the intuitive analogy between this scheme and the classical i.i.d bootstrap --- here, we treat the vertices as the units of data to be resampled --- and the widespread success of the bootstrap in the i.i.d case, vertex resampling procedures have not heretofore been rigorously analyzed\footnote{\citet{Owen-Eckles-bootstrapping-data-arrays}, drawing on \citet{Mccullagh-resampling-and-exchangeable-arrays}, consider a similar bootstrap for estimating the variance in the mean of a multi-dimensional array of real-valued random variables.}. One potential explanation for this is that when the number of resampled vertices $m$ is sufficiently large, with high probability the resampled graph $G_m^{\ast}$ will contain multiple copies of the same vertex in $G_n$. Since $G_n$ is a simple graph containing no loops, these copies will never be connected in $G_m^{\ast}$. This (non)-adjacency structure between copies of the same vertex does not reflect any underlying feature of the true generative process by which $G_n$ was formed, and therefore induces a bias between the conditional distribution of $G_m^{\ast}$ and the distribution of $G_n$, a fact remarked upon by \citep{Eldardiry-Neville-network-resampling, Levin-bootstrapping-latent-space-networks}. However, as we will see, under appropriate conditions this bias is asymptotically negligible, at least with regards to estimating the distribution of motif densities.

That being said, these conditions, such as on the maximum amount of sparsity tolerated, may be unrealistic depending on the particular problem of interest. Under stronger assumptions on the link function $w$, it is possible to more accurately estimate the graphon $h_n$, with resulting improvements to the downstream bootstrap procedure. This motivates our second estimator, which is exactly the restricted least squares histogram estimator set forth in \citet{Klopp-oracle-inequalities-for-graphons}.

\begin{Definition}[Histogram]
	\label{histogram graphon estimate}
	Fix an integer $r > 1$ which corresponds to the number of bins in the
	histogram, and a number $s \in (0,1]$ which corresponds to the maximum value the histogram estimate can take.  Define the set $\mathcal{Z}_{n,r}$ to consist of all functions which
	assign each of the $n$ nodes to one of the $r$ classes. Then we set
	the histogram estimate $\hat{h}_{\mathrm{hist}}$ of $h$ to be the least-squares estimate over functions which are piecewise-constant on partitions over the unit square, and which are bounded above by $s$. That is, for $Q = (Q_{ab}) \in R^{r \times r}, \|Q\|_{\infty} \leq s$ and $z \in \mathcal{Z}_{n,r}$, we set
	\begin{align}
	L(Q,z) & = \sum_{a,b \in [r] } \sum_{(i,j) \in z^{-1}(a) \times z^{-1}(b)} (A_{ij} - Q_{ab})^2 \\
	(\hat{Q},\hat{z}) & = \argmin_{Q, z}{L(Q,z)} \label{eqn:histogram_1}\\
	\hat{\theta}_{ij} & = \hat{Q}_{\hat{z}(i)\hat{z}(j)} \\
	\hat{h}_{hist}(u,v) & = \theta_{\lceil nu \rceil \lceil nv \rceil}
	\end{align}
	We refer to process of resampling using the histogram estimate $\hat{h}_{\mathrm{hist}}$ as the \emph{histogram bootstrap}.
\end{Definition}

The histogram estimator of the graphon is a specific case of the stochastic
block model, which itself dates back at least to
\cite{Fienberg-Wasserman-sociometric-relations,
	Holland-Lasky-Leinhardt-stochastic-blockmodels,
	Fienberg-Meyer-Wasserman-sociometric-relations}.  In such models, every node
is independently and randomly assigned to one of $r$ latent classes or
``blocks'', and edges form independently between nodes, with probabilities
depending only on the nodes' block assignments.  The histogram estimator used
here was introduced by
\citet{Klopp-oracle-inequalities-for-graphons}, though see also
\cite{Wolfe-Olhede-nonparametric-graphon-estimation, Gao-graphon-estimation, Gao-optimal-completion}. \citet{Klopp-oracle-inequalities-for-graphons} derive upper bounds on the mean-squared error of this particular histogram estimator which hold for all $Cn^{-1} \leq \rho_n \leq 1$. In our analysis of the histogram bootstrap, we will make use of these estimates on mean-squared error; interestingly, we will also make use of the fact that the histogram estimate $\hat{h}$ is itself bounded, by construction.

\paragraph{Related Work}
We have already mentioned a few suggested schemes for quantifying uncertainty of network statistics. Of these, the closest to our own approach is that of~\citet{Bhattacharyya-subsample-count-features}, who consider a pair of subsampling schemes which they show lead to consistent distributional estimates for the same types of statistics (motif densities) that we analyze. The empirical graphon bootstrap can be viewed as analogous to their uniform subsampling procedure, but with nodes sampled with replacement rather than without replacement; as mentioned previously, it is known in the i.i.d setting that resampling can lead to much more accurate estimates at small sample sizes than subsampling (see e.g.~\cite{Bickel-resampling-bootstrap}). Our theoretical results also cover general motifs, as opposed to~\citet{Bhattacharyya-subsample-count-features} who study only acyclic motifs and rings.

There also exists some other related work, which appeared after an initial version of this manuscript was made available as a preprint. \citet{Levin-Levina-bootstrapping-latent-space-networks} study a model-based bootstrap under the assumption that the observed network is a random dot product graph. In a different direction, \citet{Lunde-Sarkar-subsampling-sparse-graphons} establish the consistency of a subsampling approach for more general classes of network statistics. Finally, \citet{Zhang-edgeworth-expansion-network-moments} consider a studentized version of our empirical graphon bootstrap, and use the Edgeworth expansion to derive rates of convergence and prove higher-order correctness.

\section{Main Results}
\label{sec:main-results}

Our main pair of results establish that if one samples $G_{m}^{\ast}$ using either the empirical graphon or histogram bootstraps, then the conditional distribution
(after the right scaling and centering) of $\myP{G_m^*}$ converges in
probability to the distribution of $\myP{G_n}$, under some assumptions about
the sparsity of the graphon, the structure of the motif $R$, and --- in the case
of estimation using a histogram --- the smoothness of the graphon. For notational convenience, set $\bar{\rho}_n := P_{K_2}(\hat{h})$ to be the expected edge density of the resampled graph $G_m^{\ast}$.

\begin{Theorem} \label{empirical graphon bootstrap} Let $G_m^*$ be sampled
	from the empirical graphon bootstrap. Suppose that $w \neq 1$ on a set of 
	strictly positive (Lebesgue) measure in $[0,1]^2$. For any $p$-node motif $R$, if (i) $\int_{[0,1]^2}{ w^{4|E(R)|}(u_1,u_2) du_{1:2}} < \infty$, (ii) either $R$ is acyclic and $\rho_n = \omega(n^{-1})$ or $R$ is general and $\rho_n = \omega(n^{-\frac{1}{2p}})$,
	and (iii) $m \to \infty$ and $m = \omega(\rho_n^{-4|E(R)|})$, then
	\begin{equation} \label{empirical graphon bootstrap formula}
	\begin{aligned}
	\sup_{x}~
		& \Biggl|\mathbb{P} \left(\frac{\sqrt{m}}{\bar{\rho}_n^{|E(R)|}}
		\left(\myP{G_m^*} - \myP{\hat{h}}\right) \le x \middle| G_n \right) \\
		& - \mathbb{P} \left(\frac{\sqrt{n}}{\hat{\rho}_n^{|E(R)|}} \left(\myP{G_n} -
		\myP{h}\right) \le x \right) \Biggr| \convprob 0.
	\end{aligned}
	\end{equation}
\end{Theorem}
In practice one would approximate the conditional distribution of
$P_R(G_m^{*})$ given $G_n$ through Monte Carlo: that is, by repeatedly drawing samples of
$G_m^{\ast}$ according to the resampling procedure outlined in
Section \ref{sec:resampling-procedures}.  Theorem \ref{empirical graphon
	bootstrap} establishes that, asymptotically in $n$ and the number of
resamples $B$, such a procedure accurately approximates the distribution of
$P_R(G_n)$.

We discuss the various conditions on $w$, $\rho_n$ and $m$ in detail in Sections~\ref{sec:supporting} and~\ref{sec:proofs}, which is also where the proof of Theorem~\ref{empirical graphon bootstrap} and all our other theorems can be found. For now, we make only a few basic observations. Theorem~\ref{empirical graphon bootstrap} shows that the empirical graphon bootstrap is consistent under the relatively modest condition $\|w\|_{4|E(R)|} < \infty$. On the other hand, a more severe limitation of the empirical graphon bootstrap lies in the assumption for general motifs $R$ that $\rho_n = \Omega(n^{-1/2p})$. In contrast, the scaled and centered motif density $\rho_n^{-|E(R)|} \sqrt{n}\bigl(P_R(G_n) - P_R(h)\bigr)$ is known to converge to a non-degenerate Gaussian limit as long as $\rho_n = \Omega(n^{-2/p})$, which allows for much sparser sequences of graphs.

To obtain a consistent distributional estimate for these sparser sequences, we turn to a different graphon estimate. Under appropriate smoothness assumptions on $w$, estimators such
as $\hat{h}_{\mathrm{hist}}$ converge in a strong sense (e.g. in an $L^p([0,1]^2)$ norm) to the graphon function $h_n$; by contrast,
$\hat{h}_{\mathrm{adj}}$ converges to $h_n$ only in a weaker topology. By leveraging this stronger notion of convergence, we can prove that the
subsequent bootstrap procedure is consistent under weaker minimal conditions on the sparsity $\rho_n$.

Before focusing on the histogram estimator $\hat{h}_{hist}$, we start by considering an arbitrary graphon estimate $\hat{h}$: this could be be formed by binning, smoothing, low-rank reconstruction, etc. In Theorem~\ref{thm:convergence_distribution_estimated_graphon}, we show that so long as $\hat{h}$ satisfies a pair of general conditions, the distribution of the resampled graph $G_n^{\ast}$ will converge to that of the original graph $G_n$. 
\begin{Theorem}
	\label{thm:convergence_distribution_estimated_graphon}
	Let $\hat{h}$ be a graphon estimate, and let $G_n^{\ast}$ be a graph resampled from $\hat{h}$ as described in Section~\ref{sec:resampling-procedures}. Suppose $w \neq 1$ on a set of positive Lebesgue measure on $[0,1]^2$. For a $p$-node motif $R$, if (i) $\int_{[0,1]^2} w^{2|E(R)|}(u,v) \,du \,dv < \infty$, (ii) either $R$ is acyclic and $\rho_n = \omega(n^{-1})$ or $R$ is general and $\rho_n = \omega(n^{-2/p})$, and (iii) the graphon estimate $\hat{h}$ satisfies
	\begin{equation}
	\label{asmp:convergence_estimation_1}
	\|\hat{h}\|_{3|E(R)|} = O_p(\rho_n),\quad\textrm{and}\quad \|\hat{h} - h_n\|_{2|E(R)|} = o_p(\rho_n),
	\end{equation}
	then
	\begin{equation}
	\label{eqn:convergence_distribution_estimated_graphon}
	\begin{aligned}
	\sup_{x}~ 
		& \Biggl|\mathbb{P} \left(\frac{\sqrt{n}}{\bar{\rho}_n^{|E(R)|}} \left(\myP{G_n^*} - \myP{\hat{h}}\right) \le x \middle| G_n \right) \\ 
		& - \mathbb{P} \left(\frac{\sqrt{n}}{\hat{\rho}_n^{|E(R)|}} \left(\myP{G_n} - \myP{h}\right) \le x \right) \Biggr| \convprob 0.
	\end{aligned}
	\end{equation}
\end{Theorem}
Of course, for Theorem~\ref{thm:convergence_distribution_estimated_graphon} to be practically useful, we need an estimator $\hat{h}$ which actually achieves the notions of boundedness and convergence that are assumed in~\eqref{asmp:convergence_estimation_1}. Various difficulties arise in finding such an estimator:
\begin{itemize}
	\item {\bf Smoothness}. Although in Theorem~\ref{thm:convergence_distribution_estimated_graphon}, the only explicit assumption made on the function $w$ is that $w \in L^{2|E(R)|}([0,1]^2)$, this is insufficient to guarantee the convergence of $\hat{h} - h_n$ in the sense of~\eqref{asmp:convergence_estimation_1}. In order to obtain such a result, some additional structure must be placed on $w$. Typical assumptions are that $w$ is piecewise constant (which corresponds to the stochastic block model) or that $w$ is $\alpha$-H\"{o}lder for some $\alpha \in (0,\infty]$. For concreteness, we will stick with the assumption that $w$ is $L$-Lipschitz,
	meaning
	\begin{equation*}
	|w(u,v) - w(x,y)| \leq L\bigl(|u - x| + |v - y|\bigr),~~\textrm{for all $(u,v)$ and $(x,y) \in [0,1]^2$.}
	\end{equation*}
	\item {\bf Sparsity}. The sparsity parameter $\rho_n$ plays various roles in the hardness of graphon estimation. Intuitively, as the graph becomes sparser the whole function $h_n$ gets closer to $0$, and is thus easier to estimate; on the other hand, the available data suffers from a worse signal-to-noise ratio. Moreover, we note that as $\rho_n$ decreases the conditions in~\eqref{asmp:convergence_estimation_1} become stronger; thus it is important to have an estimator which has a fast rate of convergence in the very-sparse regime, where say $\rho_n = n^{-1} \log \log n$.
	\item {\bf Correct norm}. Assuming appropriate conditions on smoothness and sparsity, various works have considered the problem of graphon estimation. Typically these works study either the squared-error loss $\|\hat{h} - h\|_2^2$ or an in-sample analogue. However, for any motif $R$ except for the edge $R = K_2$, the norm $\|\cdot\|_{2|E(R)|}$ is a stricter norm than the $\|\cdot\|_{2}$, in the sense that $\|\cdot\|_{2} \leq \|\cdot\|_{2|E(R)|}$. So it is not the case that every graphon estimate $\hat{h}$ which accurately approximates $h$ in $L^2([0,1]^2)$ norm necessarily yields a consistent bootstrapping procedure. 
\end{itemize}
We consider the particular restricted histogram estimator of~\cite{Klopp-oracle-inequalities-for-graphons} precisely because, assuming the link function $w$ is Lipschitz, the estimator converges in a sufficiently strong norm at a sufficiently fast rate.
\begin{Proposition}[Corollary~3.6 of \cite{Klopp-oracle-inequalities-for-graphons}.]
	\label{prop:klopp}
	Suppose $w$ is $L$-Lipschitz, $\rho_n = \omega(n^{-1})$, and the histogram estimate $\hat{h}_{\mathrm{hist}}$ is computed with $r = \sqrt{n \rho_n}$ and $s = \rho_n$. Then
	\begin{equation}
	\label{eqn:klopp}
	\|\hat{h}_{\mathrm{hist}}\|_{\infty} \leq \rho_n, \quad\textrm{and}\quad \|\hat{h}_{\mathrm{hist}} - h_n\|_{2} = O_p\Bigl(\sqrt{\rho_nn^{-1}\log(n\rho_n)}\Bigr) = o_p(\rho_n).
	\end{equation}
\end{Proposition}
Using H\"{o}lder's inequality, we see that the guarantees in~\eqref{eqn:klopp} imply the conditions in~\eqref{asmp:convergence_estimation_1}. Therefore we may apply Theorem~\ref{thm:convergence_distribution_estimated_graphon}, and conclude that the histogram bootstrap yields a consistent distributional estimate.
\begin{Corollary} 
	\label{cor:convergence_distribution_histogram}
	Let $G_n^{\ast}$ be sampled from the histogram bootstrap, where the histogram estimate $\hat{h}_{\mathrm{hist}}$ is computed with $r = \sqrt{n \rho_n}$ and $s = \rho_n$. Suppose $w$ is $L$-Lipschitz, and $w \neq 1$ on a set of positive Lebesgue measure on $[0,1]^2$. For a $p$-node motif $R$, if either $R$ is acyclic and $\rho_n = \omega(n^{-1})$ or $R$ is general and $\rho_n = \omega(n^{-2/p})$, then
	\begin{equation}
	\label{eqn:convergence_distribution_histogram}
	\begin{aligned}
	\sup_{x}~
	& \Biggl|\mathbb{P} \left(\frac{\sqrt{n}}{\bar{\rho}_n^{|E(R)|}} \left(\myP{G_n^*} - \myP{\hat{h}}\right) \le x \middle|G_n \right) \\
	& - \mathbb{P} \left(\frac{\sqrt{n}}{\hat{\rho}_n^{|E(R)|}} \left(\myP{G_n} - \myP{h}\right) \le x \right) \Biggr| \convprob 0.
	\end{aligned}
	\end{equation}
\end{Corollary}
Corollary~\ref{cor:convergence_distribution_histogram} shows that when the link function $w$ is Lipschitz, the histogram bootstrap is consistent for all ranges of sparsity in which the limiting distribution of $\rho_n^{-|E(R)|} \sqrt{n}\bigl(P_R(G_n) - P_R(h)\bigr)$ is known (see~\cite{Bickel-Chen-Levina-method-of-moments}). As promised, when $R$ is a general motif, the sparsity assumptions needed for Theorem~\ref{thm:convergence_distribution_estimated_graphon} and Corollary~\ref{cor:convergence_distribution_histogram} to hold can be much weaker than those required for Theorem~\ref{empirical graphon bootstrap}. For example, suppose the motif of interest $R$ is a triangle,
i.e. $R = \bigl\{\{1,2,3\},\{(1,2),(2,3),(1,3)\}\bigr\}$ is the complete graph on $3$ nodes. In this case, for the guarantees
of Theorem \ref{empirical graphon bootstrap} to hold the sparsity must satisfy $\rho_n = \omega(n^{-1/6})$. On the other hand, assuming $w$ is Lipschitz
the conclusions of Corollary~\ref{cor:convergence_distribution_histogram}
hold whenever the sparsity is at least $\omega(n^{-2/3})$, which is much smaller.

There are certainly drawbacks to the histogram bootstrap. First of all, as already mentioned, to obtain stronger results in terms of the sparsity parameter $\rho_n$, we are forced to make stronger assumptions on the link function $w$. Additionally, we note that the estimator $\hat{h}_{hist}$ explicitly enforces sparsity in the estimate through the tuning parameter $s$, which in Theorem~\ref{thm:convergence_distribution_estimated_graphon} is set based on the (typically unknown) sparsity $\rho_n$. \cite{Klopp-oracle-inequalities-for-graphons} suggest a data-dependent way for choosing $s$, which will not affect the results of~\eqref{eqn:klopp} nor the consistency of the resulting bootstrap procedure. Finally, although the estimator of~\cite{Klopp-oracle-inequalities-for-graphons} has strong theoretical guarantees, it involves solving a combinatorial optimization problem and is not computationally feasible. We now turn to a discussion of this and other computational issues.

\paragraph{Computational considerations}
Each of our proposed bootstraps are computationally intensive procedures. Naively
computing $P_R(G_m^{\ast})$ even once requires checking $O(m^p)$ separate
subgraphs, which may be infeasible except when the sample size $n$ is modest, and the number of vertices in the motif $R$ is small.  Repeating this computation for each bootstrap sample only exacerbates these issues. The histogram bootstrap poses the added challenge that one need compute the estimate $\hat{h}_{hist}$, which as already mentioned poses a serious computational challenge.

Still, the situation is not as bleak as these considerations
suggest, for several reasons.  First, for certain graphs, such as stars and
triangles, there exist fast (indeed, sublinear time) algorithms for
approximating the subgraph count
\citep{Eden-triangle-counting,Gonen-star-counting}.  Second, the motifs whose
densities are of {\em scientific} interest in applied problems tend to be
small. Third, any graphon estimate $\hat{h}$ which satisfies~\eqref{eqn:klopp} will meet the condition~\eqref{asmp:convergence_estimation_1}, and will yield a consistent bootstrap. For instance, the spectral method studied by~\cite{Xu-spectral-methods-graphon-estimation} is a computationally reasonable alternative to the histogram estimate which---after appropriate truncation of the estimate $\hat{h}$ to satisfy $\|\hat{h}\|_{\infty} \leq \rho_n$, and extension of the estimate to $[0,1]^2$---should satisfy~\eqref{eqn:klopp}.

\paragraph{Practical implications}
To summarize, the theory developed in this section shows that our proposed empirical graphon and histogram bootstraps result in consistent estimates for the distributions of motif densities. Motif densities are statistics of direct scientific interest, for instance in biological~\citep{Milo-et-al-motifs} and social~\citep{chains-of-affection} networks. Additionally, our conclusions apply to sparse graphons, which are one way to model the sparsity often seen in real-world networks. 

Finally, the fact that our bootstraps are valid for motif
densities suggests (via the delta method) that they should also be valid for
functionals which can be expressed as well-behaved functions of motif
densities. Precisely because motif densities characterize infinite-exchangeable network distributions, this class of functionals should be quite rich, and include more complicated and scientifically interesting functionals than just motif densities. We
thus view our results as evidence that the bootstrap is a statistically reasonable off-the-shelf method of generally quantifying uncertainty in network data analysis.

\section{Simulations}
\label{sec:simulations}

Our theory kicks in only as the number of vertices $n \to \infty$. Of course we would like to be confident that our bootstraps work when $n$ is finite, and even when $n$ is
reasonably small. In this section, we provide empirical evidence supporting
this conclusion. We consider various
graphons $h_n$ and motifs $R$, and show that confidence intervals formed using either of our proposed bootstraps typically have close to nominal
coverage when $n$ is even moderately large. Moreover, the interval widths shrink quickly
and at comparable rates for each bootstrap procedure.

\paragraph{Setup}
We will consider three separate simulation setups. In each, the graph $G_n$
will be sampled as in~\eqref{eqn:graphon_sample}, and the
setups will be distinguished only by different choices of the link function
$w$. Within each setup, we will vary $n = 25,50,\ldots,400$ and $\rho_n = .02,.1$ and $.25$, to investigate the effects of both sample size and sparsity. We
will always draw $100$ resampled graphs
$(G_n^{\ast})^{(1)},\ldots,(G_n^{\ast})^{(100)}$ with resulting motif densities
$P_R((G_n^{\ast})^{(1)}), \ldots, P_R((G_n^{\ast})^{(100)})$, and take the
empirical distribution
\begin{equation*}
\frac{1}{100}\sum_{b = 1}^{100} \frac{\sqrt{n}}{\bar{\rho}_n^{|E(R)|}}\Bigl(P_R\bigl((G_n^{\ast})^{(b)}\bigr) - P_R(\hat{h})\Bigr)
\end{equation*}
as a Monte Carlo estimate of the (centered and scaled) conditional distribution
of $P_R(G_n^{\ast})$. Letting $1 - \alpha$ be the nominal coverage, we take
$\widehat{q}_{\alpha/2}$ and $\widehat{q}_{1 - \alpha/2}$ to be the $\alpha/2$
and $1 - \alpha/2$ quantiles of this Monte Carlo estimate. Our final confidence
interval is then given by
\begin{equation*}
\widehat{I}_R(G_n) = \biggl(P_R(G_n) + \widehat{q}_{\alpha/2}\frac{\widehat{\rho}_n^{|E(R)|}}{\sqrt{n}}, P_R(G_n) + \widehat{q}_{1 - \alpha/2}\frac{\widehat{\rho}_n^{|E(R)|}}{\sqrt{n}}\biggr)
\end{equation*}
Our theory establishes that, under appropriate conditions on the link function $w$, the true coverage of $\widehat{I}_R(G_n)$ should approach $1 - \alpha$ as $n \to \infty$
(up to error induced by the Monte Carlo approximation). To examine finite
sample behavior, we repeat the above setup over $N = 1000$ draws of the graph
$G_n$ --- and resulting intervals $\widehat{I}_R(G_n)$ --- to obtain an estimated
coverage and average interval width for each value of $n$ and $\rho_n$, and for
each of our two bootstraps. It is these quantities which we examine below.

For the empirical graphon bootstrap, we implemented our own code.  To compute
a histogram estimator, we use the \texttt{blockmodels} package
\citet{Leger-blockmodels-package}, which approximates the least squares
estimator using a variational EM procedure, and chooses the optimal number of
blocks by cross-validation.  Note that this deviates slightly from our
theoretical definition of the histogram estimator, but the deviation is in the
direction of being more computationally tractable, i.e., not having to solve a
difficult combinatorial optimization problem.

\paragraph{Simulation 1: Gaussian Link Function}
\label{subsec:gaussian_link_simulation}

In our first simulation, we let 
\begin{equation*}
w(u,v) \propto \exp(-25 \|u - v\|^2/2)
\end{equation*}
be a Gaussian link function with bandwidth $1/5$ (the constant of
proportionality is chosen so that $\int w(u,v) \,du \,dv = 1$). In Figure
\ref{fig:gaussian_coverage}, we see coverage plotted against sample size, for
the motifs $R = \{(1,2),(2,3), (1,3)\}$ (triangle), $R = \{(1,2),(2,3)\}$
(two-star), $R = \{(1,2),(2,3),(3,4),(4,1)\}$ (four-cycle), and $R =
\{(1,2),(1,3),(1,4)\}$ (claw). The lines are colored according to the
resampling procedure used, and marked according to whether $\rho_n = .02, .1$ or
$.25$.

\begin{figure}[!htb]
	\centering
	\includegraphics[scale = .375]{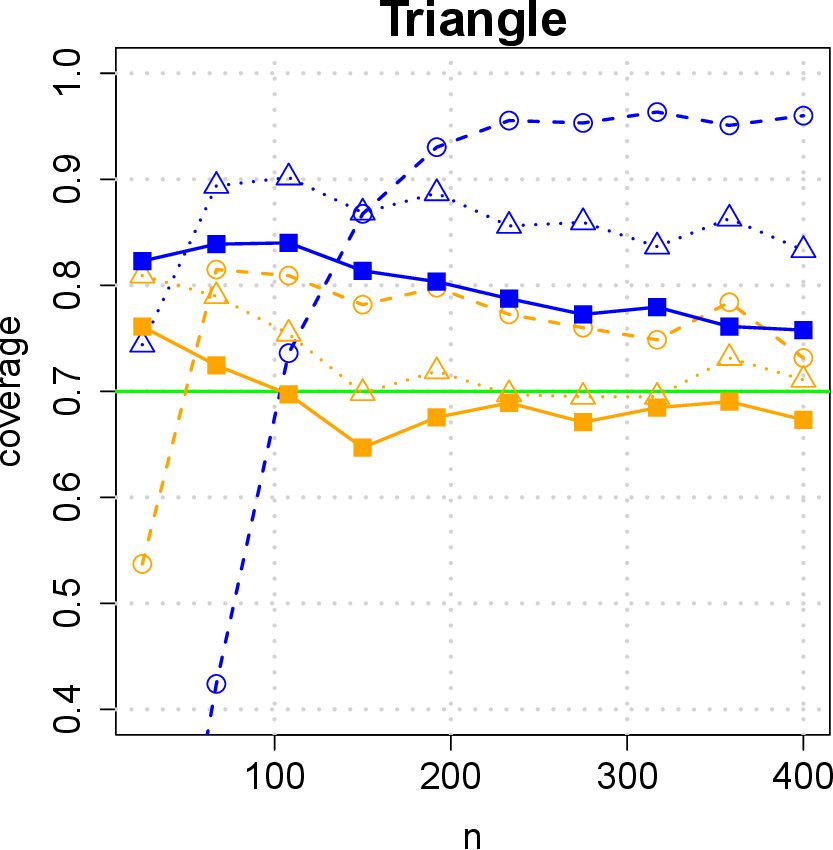}
	\includegraphics[scale = .375]{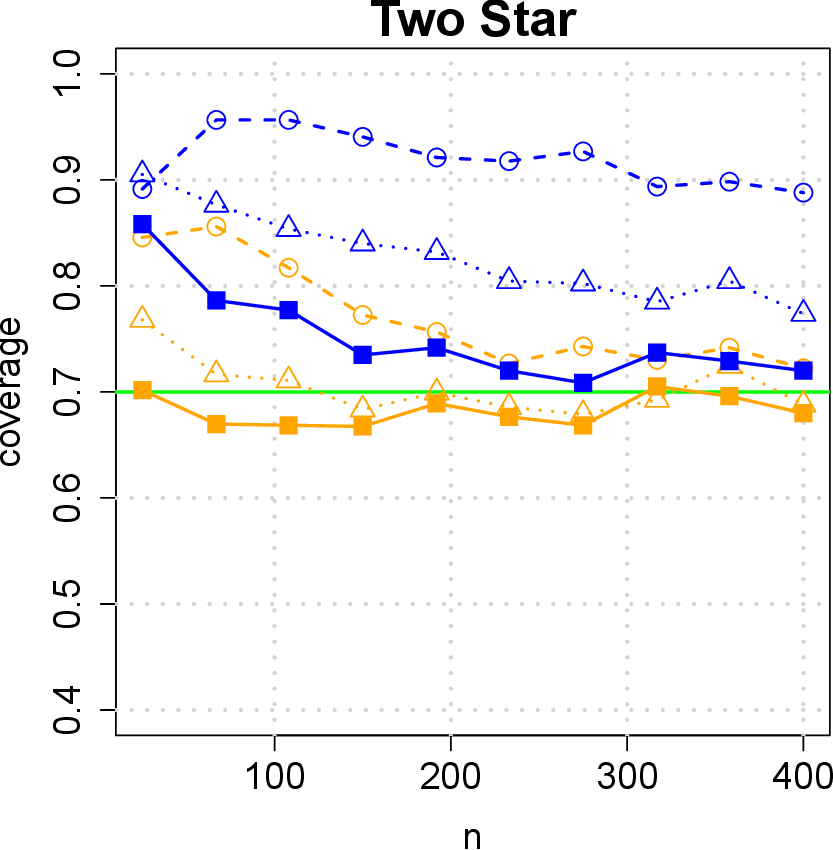}
	\includegraphics[scale = .375]{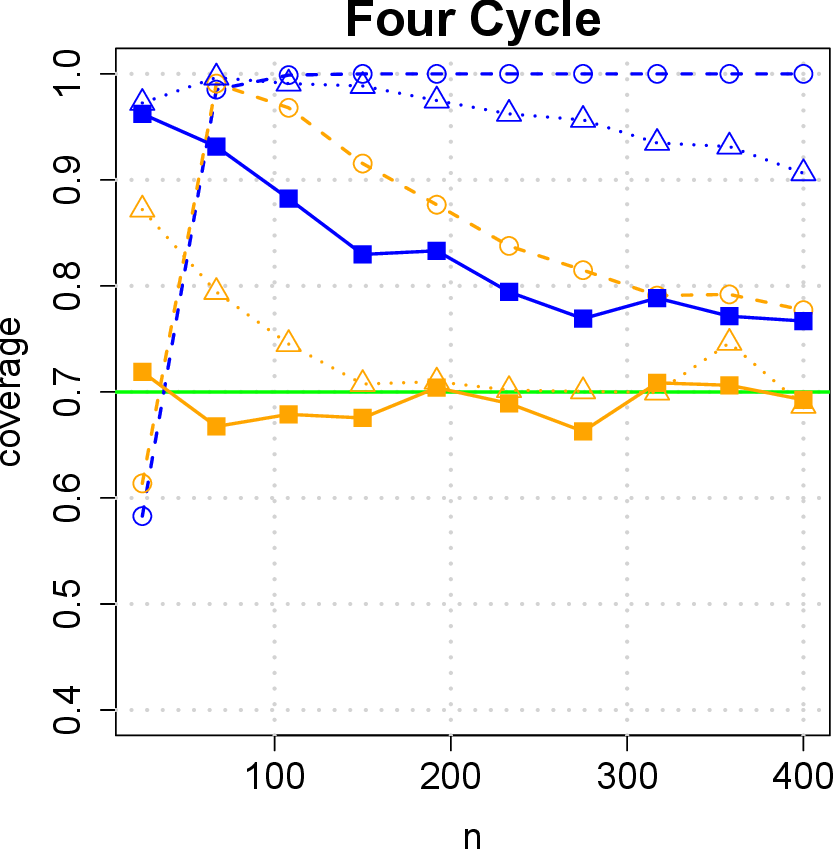}
	\includegraphics[scale = .375]{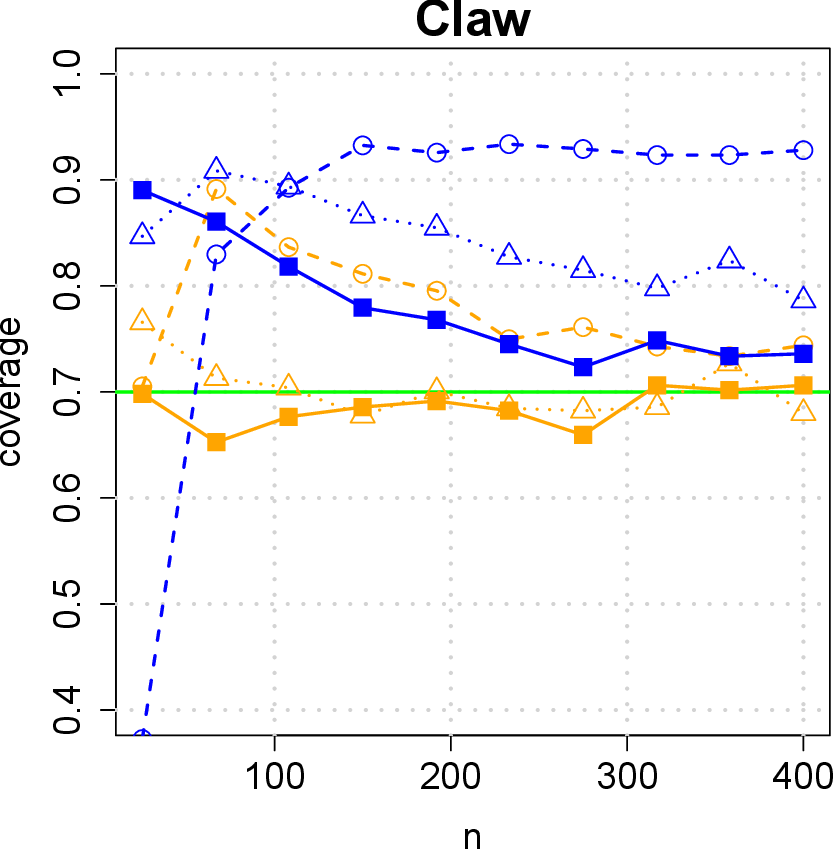}
	\caption{Coverage of bootstrapping methods as a function of sample size for Gaussian link function. Nominal level $1 - \alpha = .7$ denoted by horizontal green line. The empirical graphon procedure is shown in blue, and the histogram graphon procedure in orange. Dashed lines and circular symbols correspond to $\rho_n =.02$; dotted lines and triangles to $\rho_n = .1$; filled lines and squares to $\rho_n = .25$.}
	\label{fig:gaussian_coverage}
\end{figure}

When $\rho_n = .25$, for all of the aforementioned motifs except the four-cycle and for all $n > 300$, the estimated coverage is within $.05$ of nominal coverage for both bootstraps; it is therefore fair to say that both of our bootstraps ``work'' reasonably well in this setting. When $\rho_n = .1$, the story is more mixed. Resampling using the histogram estimate still gives approximately nominal coverage, but estimated coverage of the empirical bootstrap deviates substantially from nominal coverage. That being, even for the empirical bootstrap estimated coverage appears to be tending towards $1 - \alpha$ as $n$ increases, as we would expect. When $\rho_n = .02$, even the histogram estimator is not particularly accurate for the four-cycle. All of this is in line with our theory, which requires suitable lower bounds on $\rho_n$ as a function of $n$ for both bootstraps, with the lower bound being more stringent for the empirical graphon than for the histogram estimator, and for motifs with more nodes and edges such as the four-cycle.

\begin{figure}[!htb]
	\centering
	\includegraphics[scale = .375]{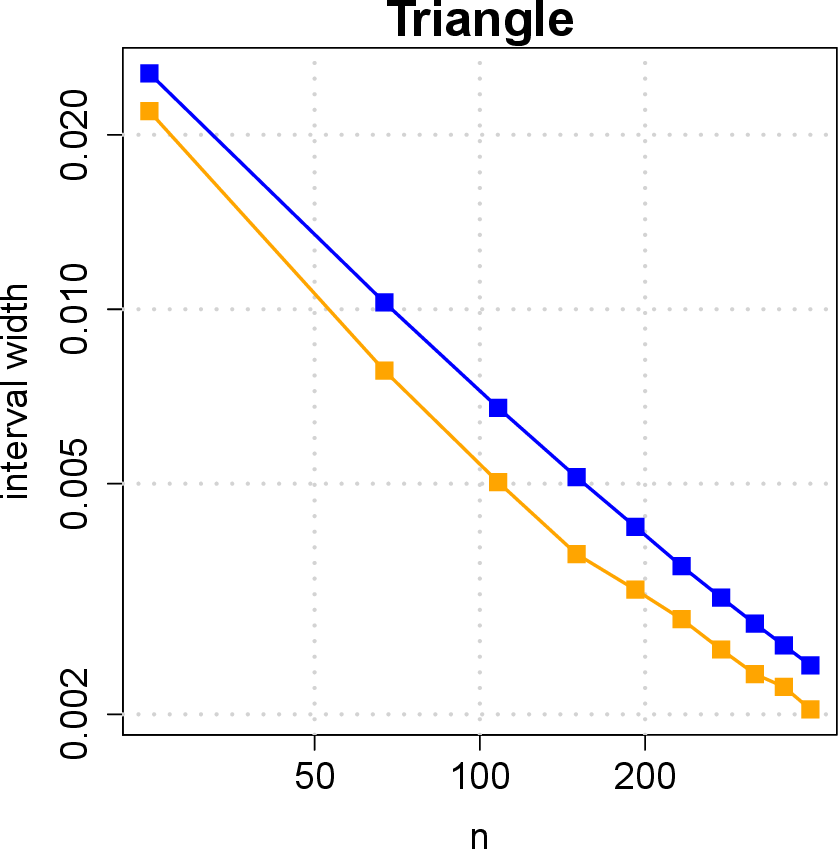}
	\includegraphics[scale = .375]{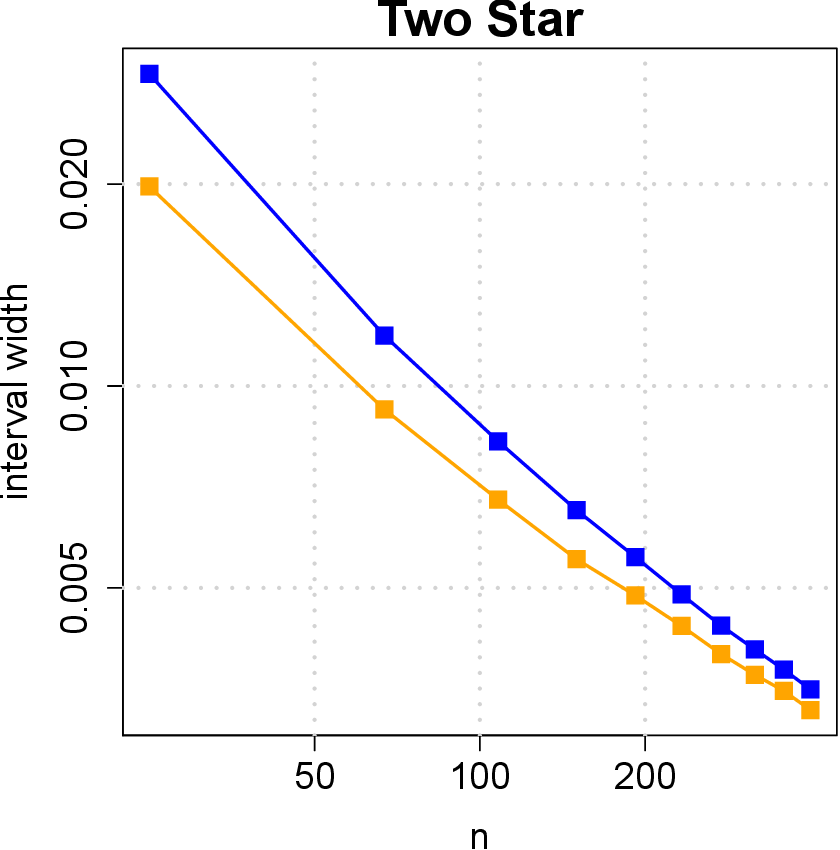}
	\includegraphics[scale = .375]{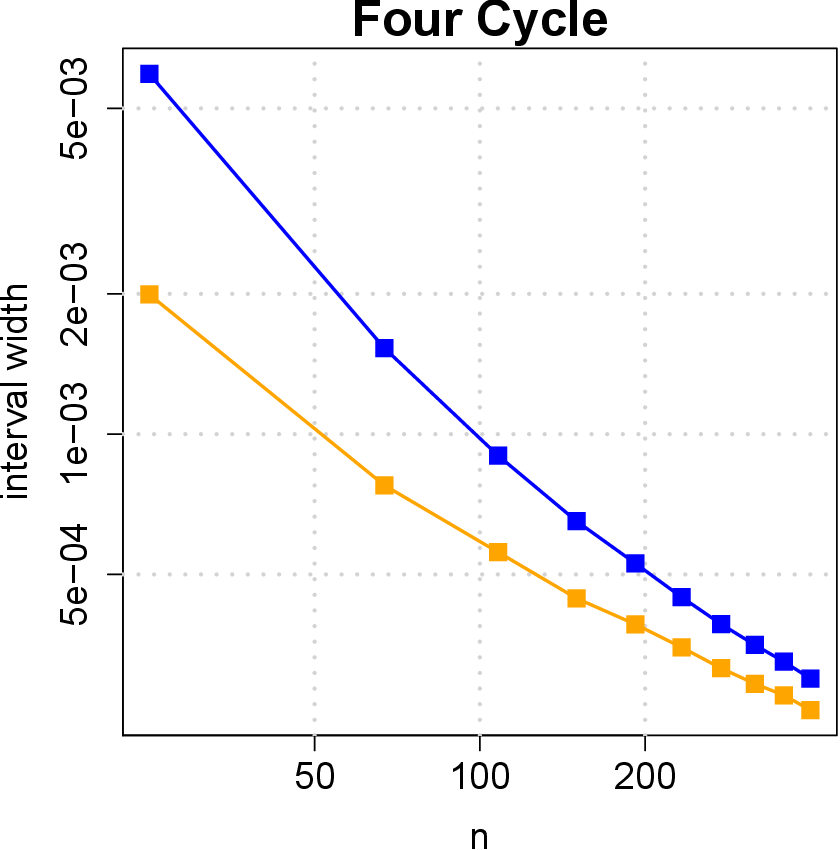}
	\includegraphics[scale = .375]{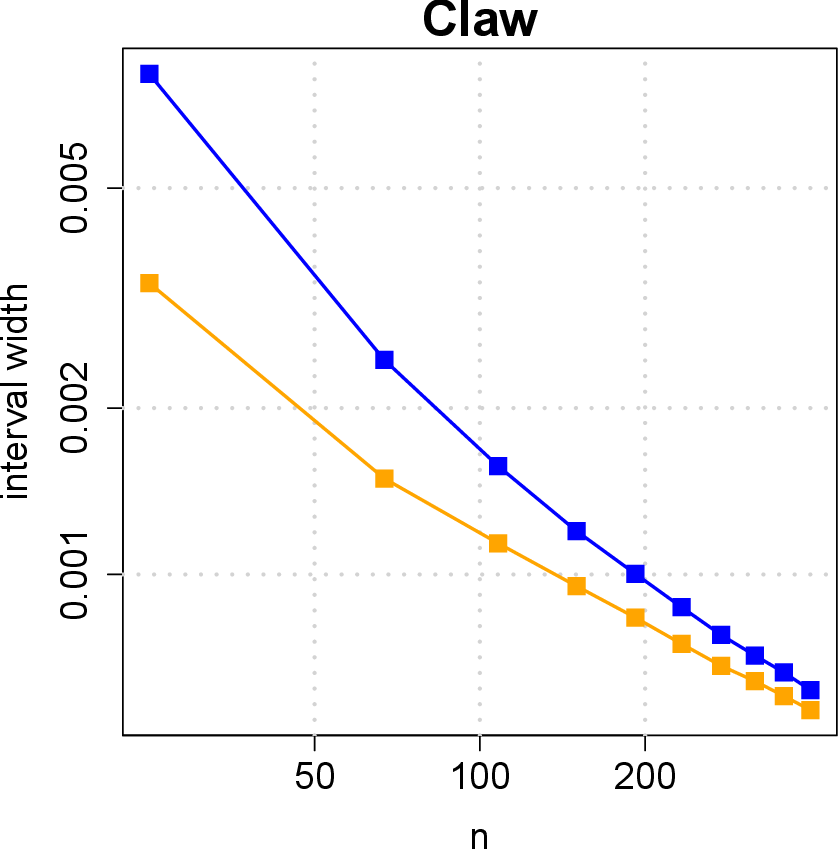}
	\caption{Interval width of bootstrapping methods as a function of sample size for Gaussian link function. Colors and symbols are the same as in Figure~\ref{fig:gaussian_coverage}.}
	\label{fig:gaussian_interval_width}
\end{figure}

Figure \ref{fig:gaussian_interval_width} shows the average interval widths for each of the four aforementioned motifs when $\rho_n= .2$, plotted on a log-log scale. We see that, in addition to having better coverage, the histogram bootstrap  has narrower confidence intervals than the empirical graphon bootstrap. However, as $n$ grows the difference between the two decreases, and the slopes become more and more similar, converging towards the $1/\sqrt{n}$ rate we would expect from Theorems~\ref{empirical graphon bootstrap} and~\ref{thm:convergence_distribution_estimated_graphon} (since $\rho_n = .25$ does not change with $n$.)

\paragraph{Simulation 2: Stochastic Block Model}

In our second simulation, we take
\begin{equation*}
w(u,v) \propto
\begin{cases}
1,~~\textrm{if $u \leq .5, v \leq .5$} \\
.8,~~\textrm{if $u > .5, v > .5$} \\
.25,~~\textrm{otherwise;}
\end{cases}
\end{equation*}
in other words we sample $G_n$ from a stochastic block model with two blocks.

Figure \ref{fig:sbm_coverage} shows estimated coverage as a function of sample size. The conclusions we draw are similar in spirit to those drawn from our first experiment. Both bootstrap methods appear to be approaching $1 - \alpha$ coverage as $n$ increases, when the density parameter $\rho_n$ is sufficiently large (although we note that for the empirical graphon bootstrap, the difference between nominal and actual coverage is larger than it was for the Gaussian link function). The histogram bootstrap has close to nominal coverage across all values of $\rho_n$, but after all this is a case where the true model is itself a histogram, and we should not be surprised that the histogram bootstrap performs quite well. 

\begin{figure}[!htb]
	\centering
	\includegraphics[scale = .375]{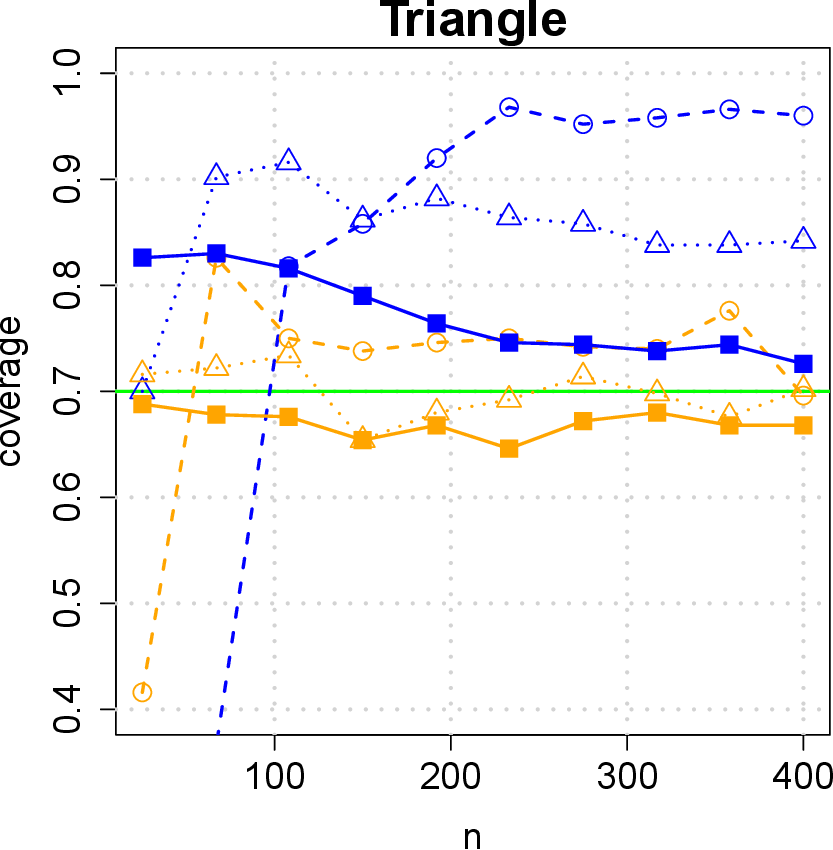}
	\includegraphics[scale = .375]{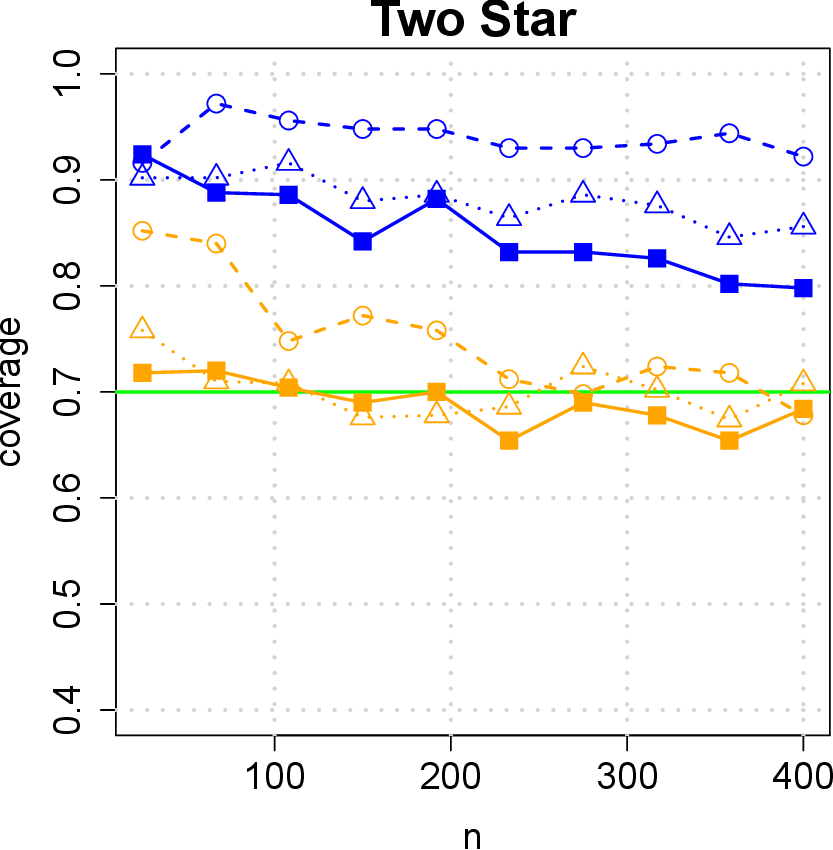}
	\includegraphics[scale = .375]{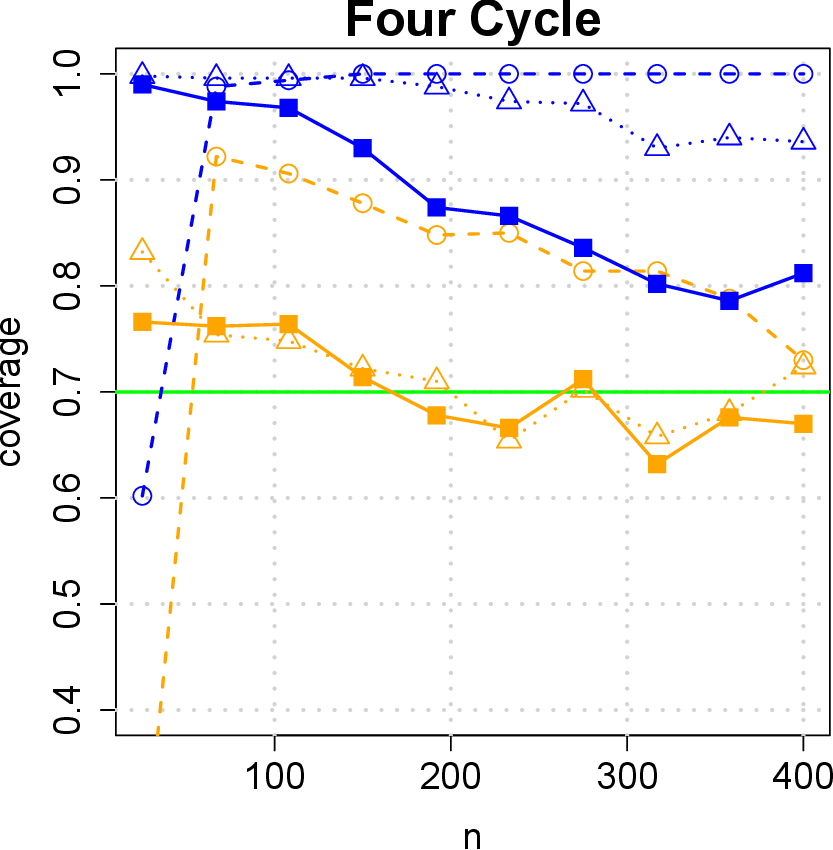}
	\includegraphics[scale = .375]{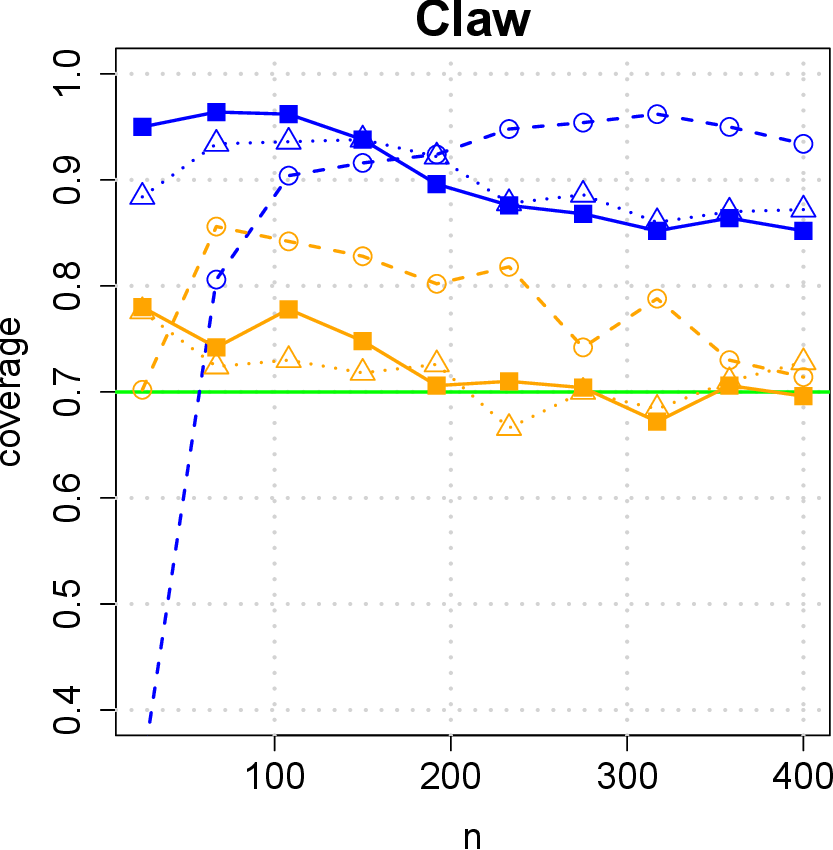}
	\caption{Coverage of bootstrapping methods as a function of sample size for stochastic block model. Colors and symbols are the same as in Figure \ref{fig:gaussian_coverage}.}
	\label{fig:sbm_coverage}
\end{figure}

Figure \ref{fig:sbm_interval_width} shows the average confidence interval width as a function of $n$ when $\rho_n= .25$. Again the conclusions are similar to those of our first experiment. In particular, while resampling using a histogram graphon estimate results in narrower average confidence intervals than resampling using the empirical graphon, the difference between the two decreases as $n$ grows, and the rate at which interval width shrinks tends towards the expected $1/\sqrt{n}$ rate.

\begin{figure}[!ht]
	\centering
	\includegraphics[scale = .375]{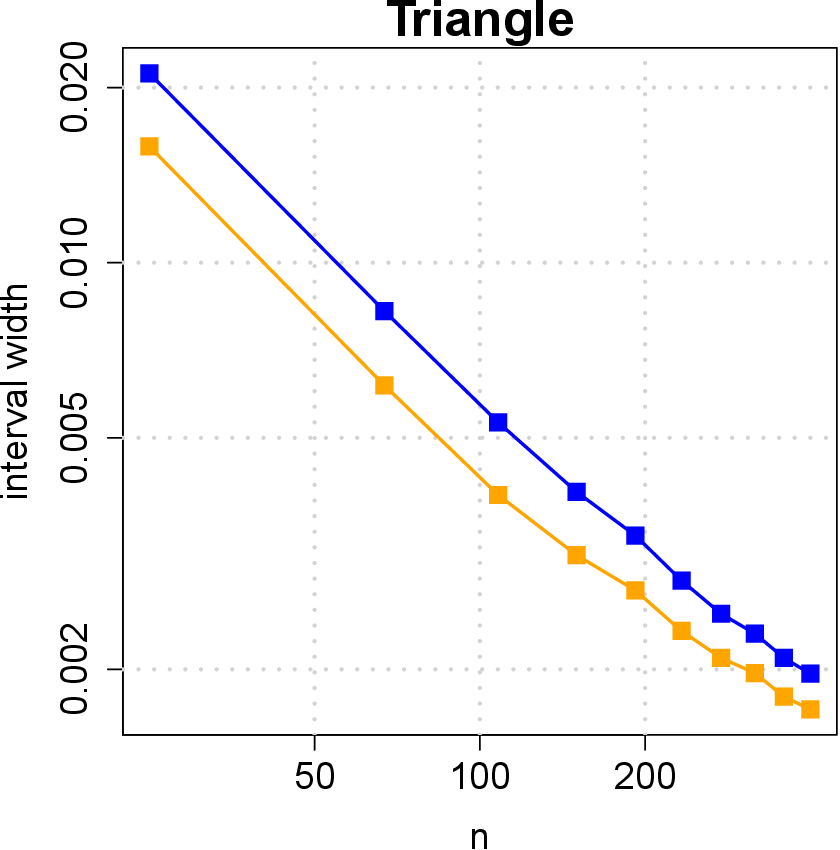}
	\includegraphics[scale = .375]{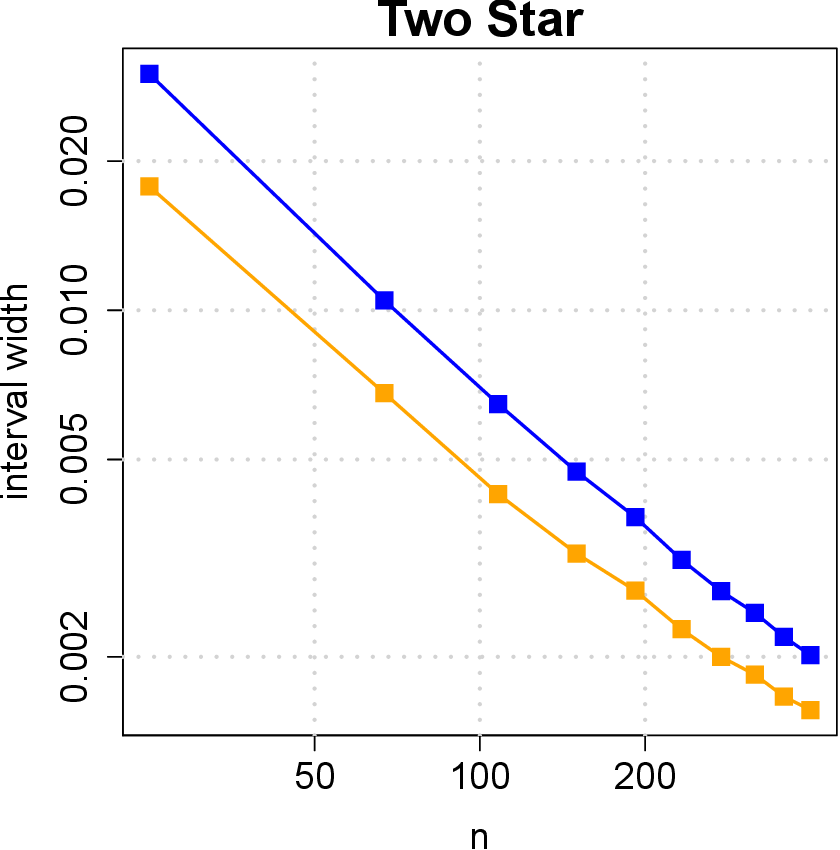}
	\includegraphics[scale = .375]{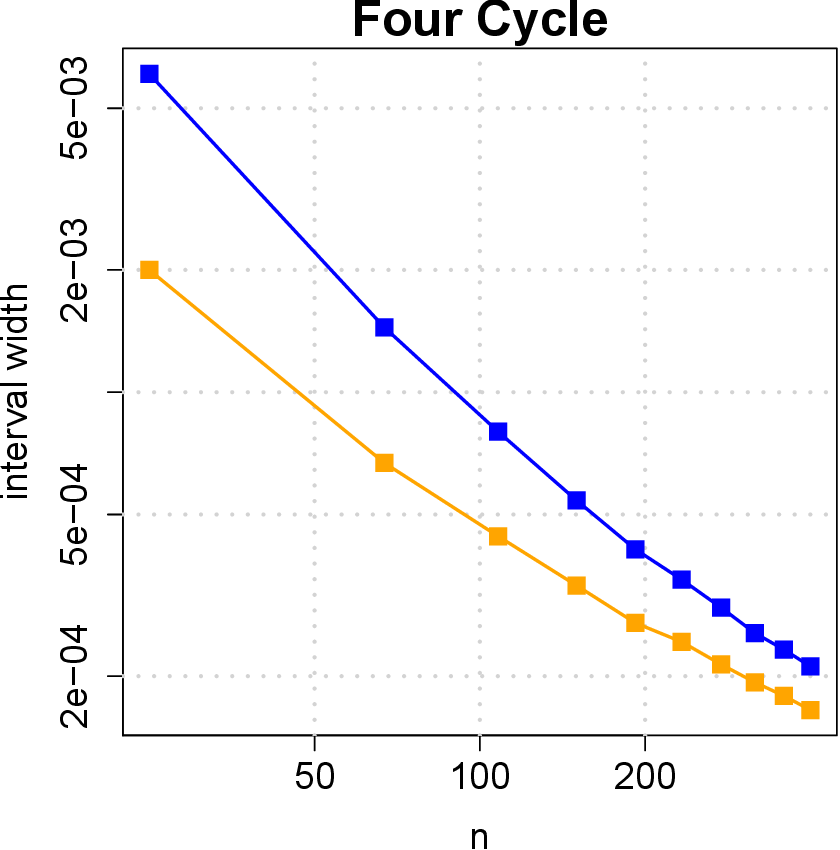}
	\includegraphics[scale = .375]{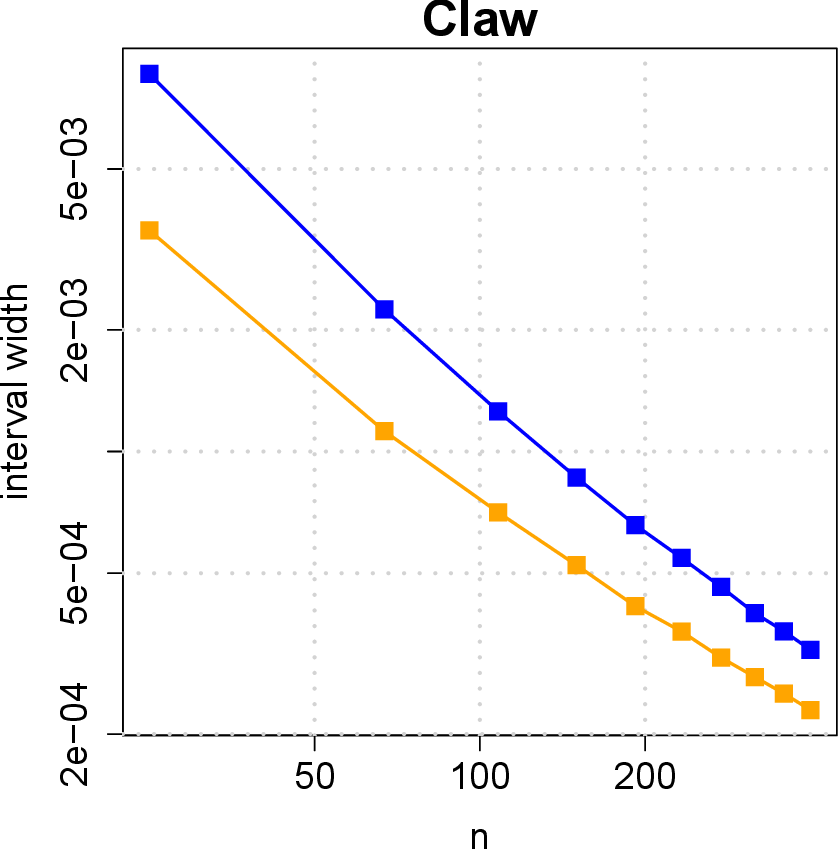}
	\caption{Interval width of bootstrapping methods as a function of sample size for stochastic block model. Colors and symbols are the same as in Figure \ref{fig:gaussian_interval_width}.}
	\label{fig:sbm_interval_width}
\end{figure}

\paragraph{Simulation 3: Horseshoe Link Function}

In our third and final simulation, we let
\begin{equation*}
w(u,v) \propto \exp\bigl(-200(u-v^2)^2/2\bigr) + \exp\bigl(-200(v-u^2)^2/2\bigr).
\end{equation*}
We call this the horseshoe link function, since when $w$ is visualized using a
heat map it looks like a horseshoe.\footnote{We have borrowed this graphon from
	\citet{Wang-network-comparisons}, where it was used as a challenging example
	for graphon estimation.} Due to its odd shape, $w$ is poorly approximated by
a histogram of any reasonable binwidth. Relatedly, the Lipschitz condition
required for Corollary~\ref{cor:convergence_distribution_histogram} is satisfied only for a very large Lipschitz constant $L$. On the other hand, our theory for the empirical graphon
does not depend on the smoothness of the function $w$, and so it applies
equally as well to the horseshoe link function as to the previous --- smoother
--- graphons considered in our first two simulations.

Such is the theoretical state of affairs; Figure \ref{fig:horseshoe_coverage}
shows the empirical reality.  In contrast to the previous two simulations, when
$\rho_n= .25$ it is now the empirical graphon bootstrap which has closer to
nominal coverage, consistently across the different possibilities for $R$ and
$n$. (When $\rho_n= .1$ or $\rho_n= .02$, resampling from a histogram estimate
still results in coverage closer to $1 - \alpha$).  We conclude that for
sufficiently non-smooth link functions $w$ and large values of $\rho_n$, the
empirical graphon may outperform the histogram estimate as a bootstrap
procedure.  Figure \ref{fig:horseshoe_interval_width} shows that average
interval widths are similar for both resampling procedures, reiterating the
takeaway message of Simulations 1 and 2.

\begin{figure}[!ht]
	\centering
	\includegraphics[scale = .375]{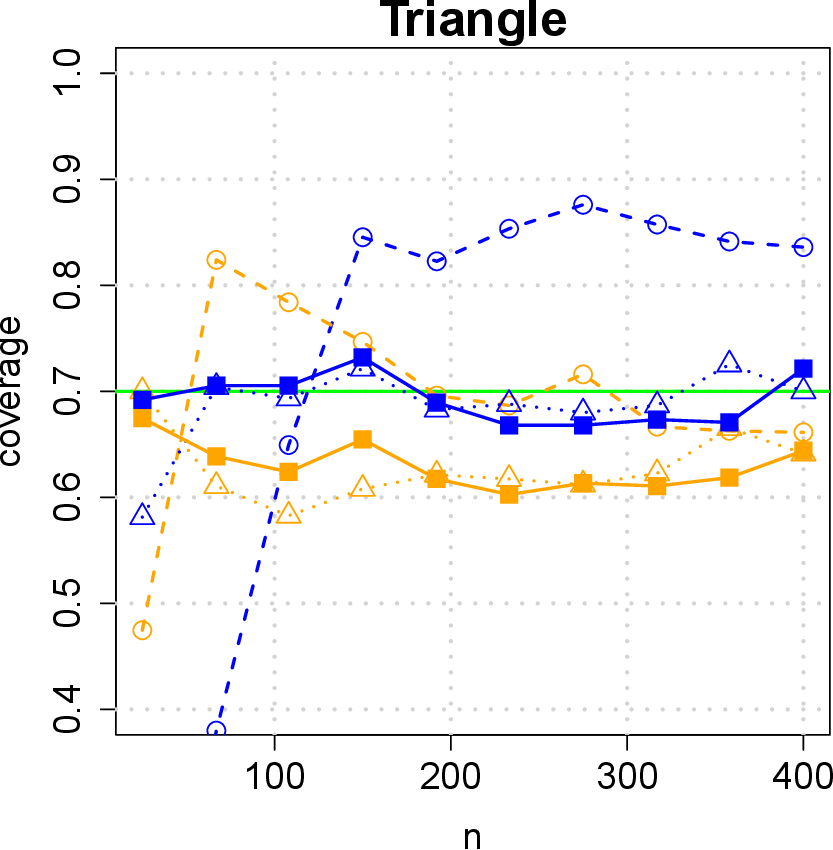}
	\includegraphics[scale = .375]{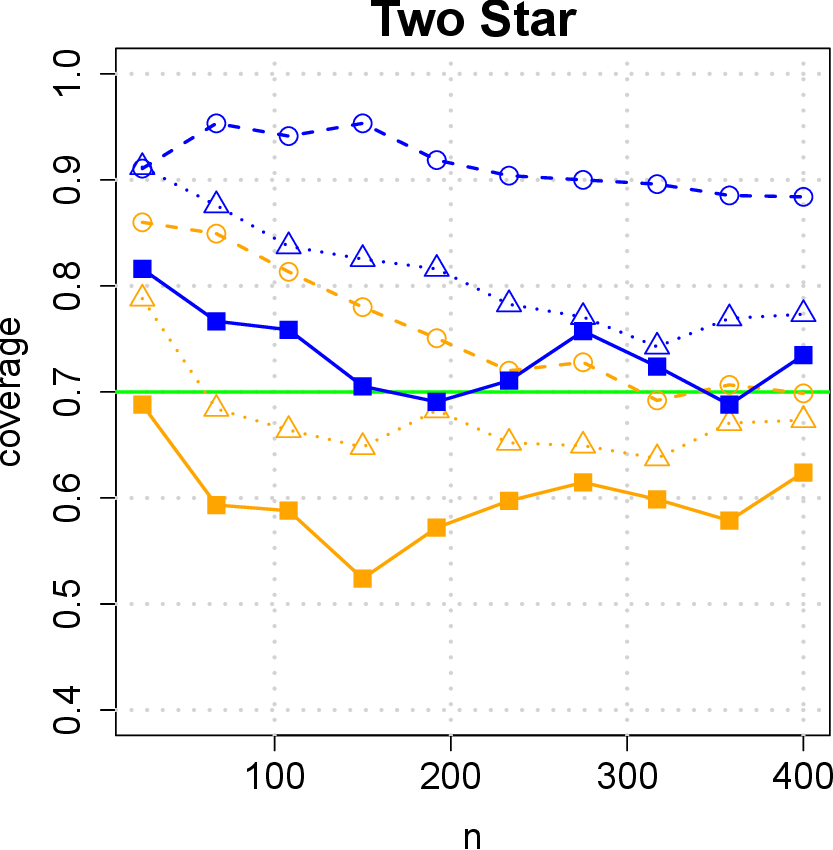}
	\includegraphics[scale = .375]{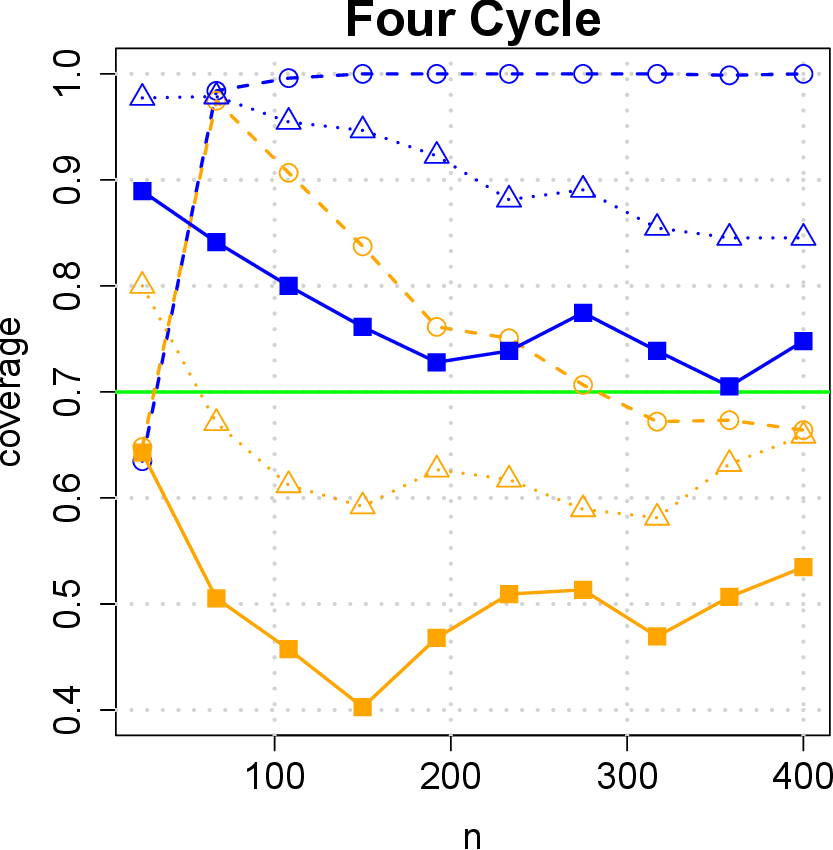}
	\includegraphics[scale = .375]{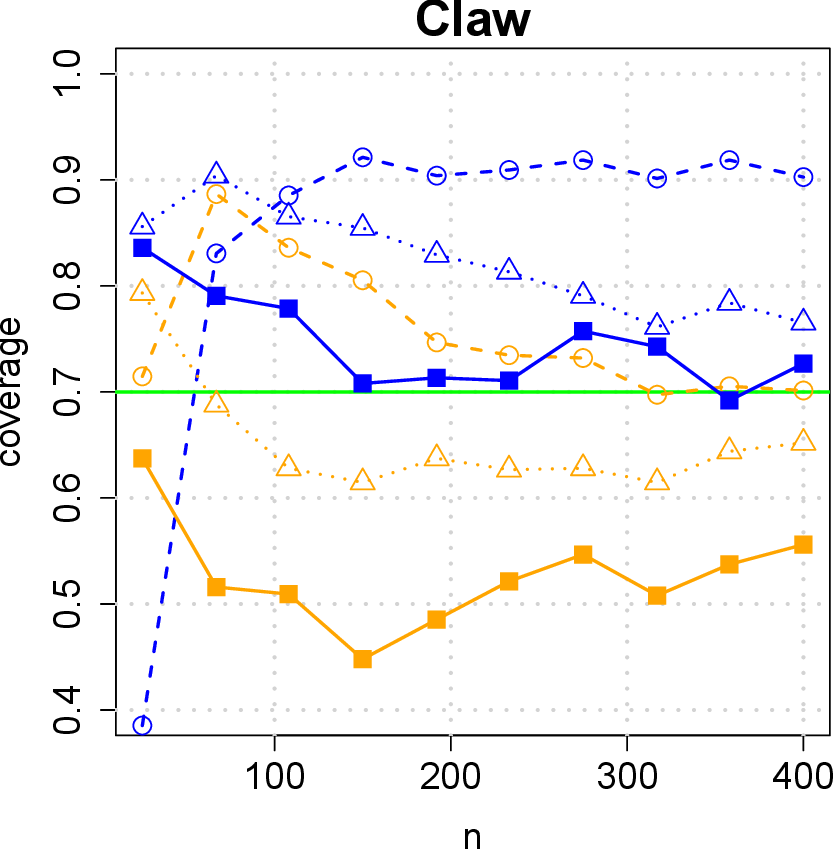}
	\caption{Coverage of bootstrapping methods as a function of sample size for horseshoe graphon. Nominal level $1 - \alpha = .7$ denoted by green line. Colors and symbols are the same as in Figure \ref{fig:gaussian_coverage}.}
	\label{fig:horseshoe_coverage}
\end{figure}

\begin{figure}[!ht]
	\centering
	\includegraphics[scale = .375]{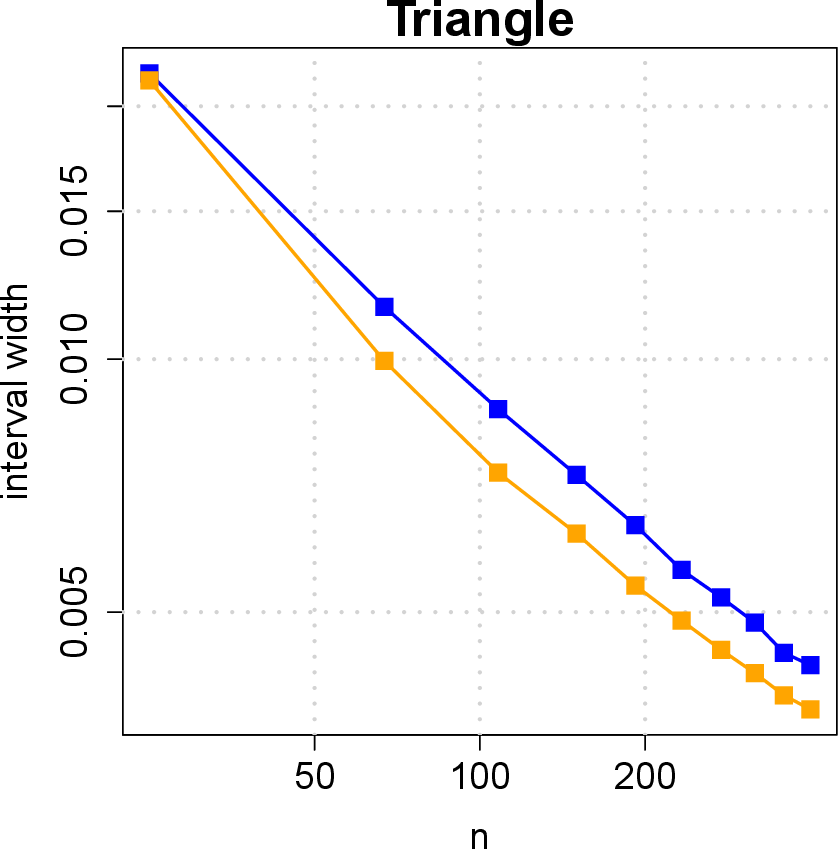}
	\includegraphics[scale = .375]{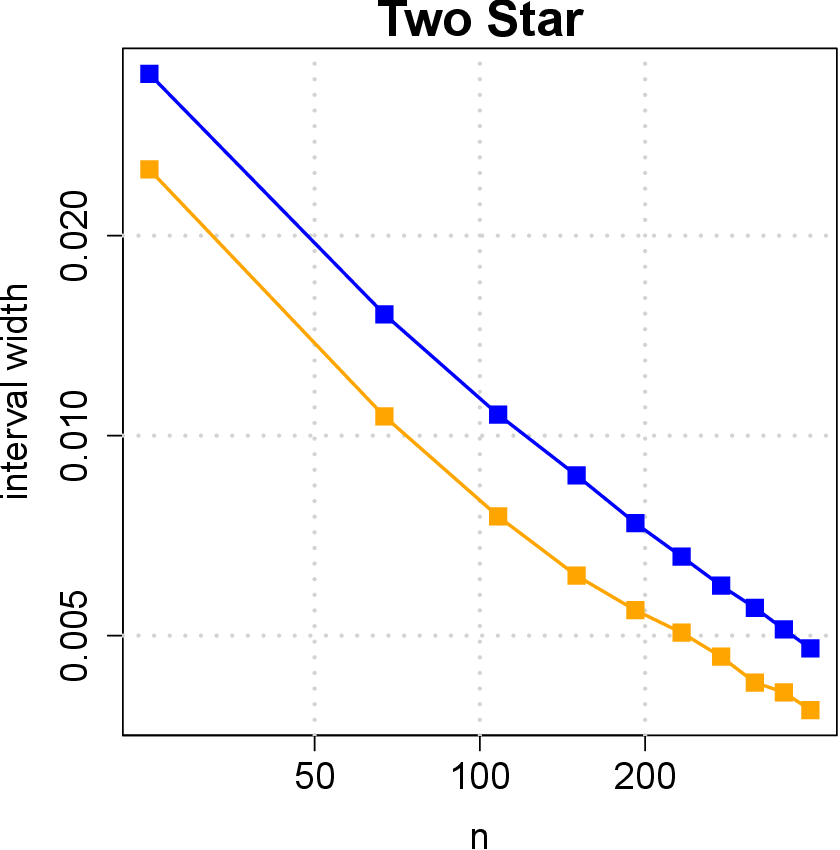}
	\includegraphics[scale = .375]{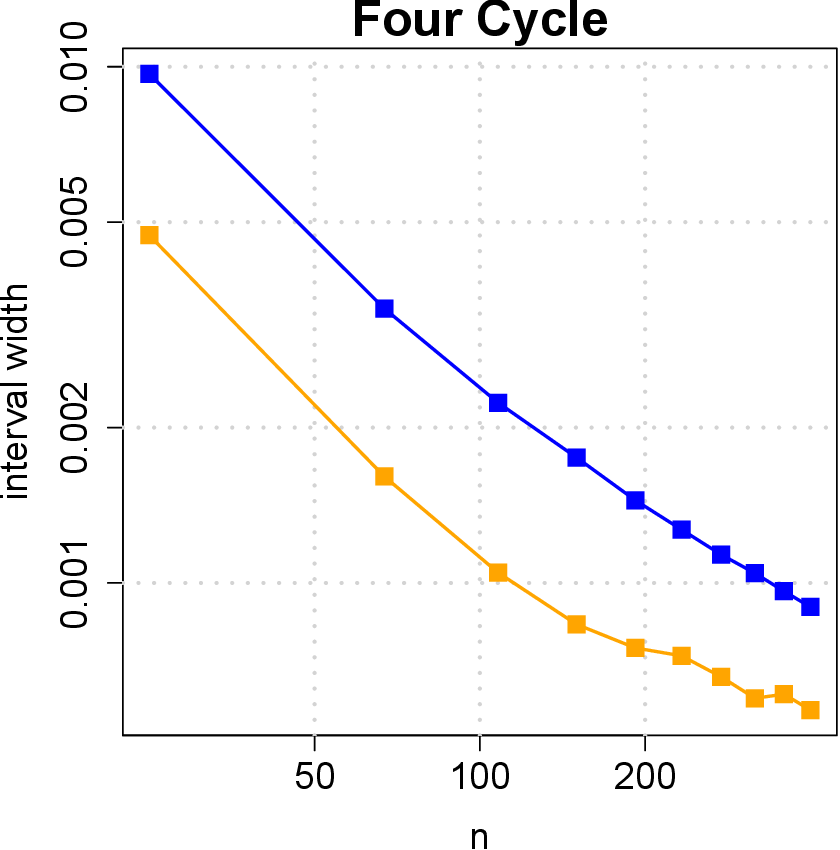}
	\includegraphics[scale = .375]{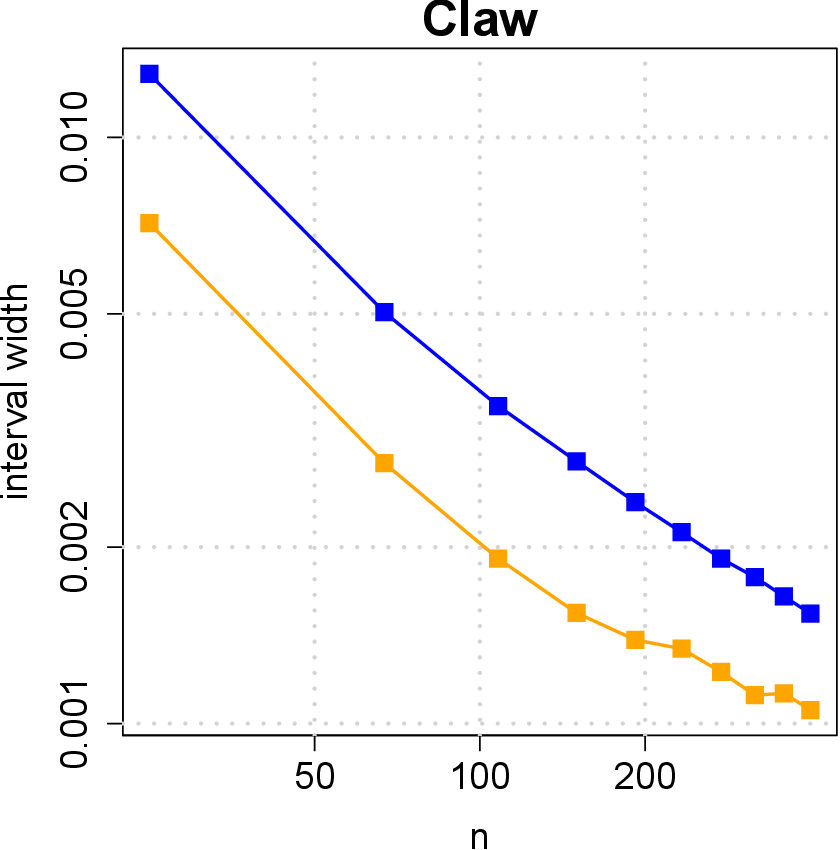}
	\caption{Interval width of bootstrapping methods as a function of sample size for horseshoe graphon. Colors and symbols are the same as in Figure \ref{fig:gaussian_interval_width}.}
	\label{fig:horseshoe_interval_width}
\end{figure}

Taken as a whole, our simulations demonstrate that even for networks with a
moderate number of vertices, both of our bootstraps result in reasonable
estimates of the motif density distribution, across different choices of
graphon function and motif. In particular, confidence intervals formed using
these bootstraps have approximately nominal coverage when $\rho_n$ is
sufficiently large, as suggested by our asymptotic theory, and interval widths
shrink rapidly (approaching the anticipated $1/\sqrt{n}$ rate for fixed
$\rho_n$) as $n$ grows.

\section{Supporting Propositions and Lemmas}
\label{sec:supporting}

As mentioned earlier, the proofs for Theorems~\ref{empirical graphon bootstrap} and~\ref{thm:convergence_distribution_estimated_graphon} each rely on a series of 3 approximations: first, that
the distribution of a scaled version of $\myP{G_n}$ is close to Gaussian; second, that
the conditional distribution of a scaled version of $\myP{G_m^*}$ is close to
Gaussian; and third, that the variances of both distributions are close. Propositions
\ref{convergences of myP} and \ref{convergence of variance} formalize the first
and third of these assertions (the second assertion is covered in the proof of
the theorems.)

In order to state Proposition \ref{convergences of myP} , we will need an
expression for the (normalized) variance of $\myP{G_n}$. Intuitively, this will
involve taking expectation over the product of indicator functions, of the form
$\Expect{\mathbb{I} \{G_n(\iset) \simeq R \}\mathbb{I} \{G_n(\jset) \simeq R \}}$. 
As we might expect, this quantity can be related to $\myP[W]{h}$ for a set of motifs $W$. 
The intuition is that the event of seeing the same motif on two different subgraphs corresponds to seeing one of
several particular motifs on the union subgraph.  We will call the set of these
motifs $W$ the ``merged copy set'', because it consists of motifs formed by
taking two copies of $R$ and merging some of their vertices.

\begin{Definition} \label{merged copy set} The \textbf{merged copy set} of a
	motif $R$ on $k$ vertices, denoted $MC(R,k)$, is the set of all $W$ such that
	\begin{equation}
	MC(R,k) = \left\{W: |V(W)| = k, \exists \iset,\jset \text{ s.t. } \iset \cup \jset = 1:k, W(\iset) \simeq R, W(\jset) \simeq R \right\}
	\end{equation}
\end{Definition}

Lemma \ref{joint sum over motifs}, next, relates the double sum present in
$\Expect{\myP{G_n}^2}$ to summing over motifs in the various merged copy sets.
The merged copy sets tells us which motifs we need to sum over, but not how
many times we'll need to consider each motif. Luckily, this depends only on
${n \choose k}$ and a combinatorial factor, itself a function of $R$
and $W$, which counts how many ways $W$ can be formed by merging two copies of
$R$.\footnote{For example, if $R = K_2$, the single edge between two nodes,
	there are four different ways we can merge two copies of $R$ to get a
	2-star.}.

\begin{Lemma} \label{joint sum over motifs} For any $p$-node motif $R$,
	\begin{equation} \label{eqn: joint sum over motifs}
	\sum_{\iset,\jset \in S_p(n)}{ \Expect{ \mathbb{I}\{G_{n}(\iset) \simeq R \} \mathbb{I} \{G_{n}(\jset) \simeq R \} }} = \sum_{k = p}^{2p}{ {n \choose k} \sum_{W \in MC(R,k)}{ C_R(W) \myP[W]{h}}},
	\end{equation}
	where $C_R(W)$ counts the number of ways of forming $W$ by merging two copies
	of $R$.
\end{Lemma}

\begin{Lemma} \label{var_of_p} Define
	\begin{equation}
	\sigma^2_R(h_n) \equiv \Var{\frac{\sqrt{n}}{\rho_n^{|E(R)|}} \myP{G_n}}.
	\end{equation}
	Then,
	\begin{align*}
	\sigma^2_R(h_n) & = \frac{n}{(\rho_n^{|E(R)|} {n \choose p} |p!N(R)|)^2} \sum_{k=p}^{2p-1} {n \choose k} \sum_{S \in MC(R,k)} C_R(S) \myP[S]{h_n} \\ & \quad - \left(1-\frac{{n-p \choose p}}{{n \choose p}} \right) \myNormP{h_n}^2 .
	\end{align*}
	Moreover, if
	$\int_{[0,1]^2}{ w^{2|E(R)|}(u_1,u_2) du_{1:2}} < \infty$, and either $R$ is acyclic and $\rho_n = \omega(n^{-1})$ or $R$ is general and $\rho_n = \omega(n^{-2/p})$, then
	\begin{equation}
	\sigma^2_R := \lim_{n \to \infty}{\sigma^2_R(h_n)} < \infty.
	\end{equation}
\end{Lemma}

Now that we know that $\rho_n^{-|E(R)|}$ is the appropriate normalization, we
can state a central limit theorem from~\cite{Bickel-Chen-Levina-method-of-moments}  for the empirical subgraph densities.

\begin{Proposition}[Theorem~1 of \citep{Bickel-Chen-Levina-method-of-moments}] \label{convergences of myP} Suppose $w \neq 1$ on a subset of $[0,1]^2$ with positive Lebesgue measure. For a $p$-node motif $R$, if (i) $\int_{[0,1]^2}{  w^{2|E(R)|}(u_1,u_2) du_{1:2}} < \infty$ and (ii) either $R$ is acyclic and $\rho_n = \omega(n^{-1})$, or $R$ is general and $\rho_n = \Omega(n^{-\frac{2}{p}})$ then
	\begin{equation} \label{convergence in dist of myP} 
	\sqrt{n} \rho_n^{-|E(R)|}\Bigl(\myP{G_n} - \myP{h}\Bigr) \convdist \mathcal{N}(0,\sigma^2_R)
	\end{equation}
	Also
	\begin{equation} \label{convergence in prob of myP} \myNormP{G_n} \convprob
	\myNormP{h}
	\end{equation}
\end{Proposition}
We reiterate that the requirement $\rho_n = \omega(n^{-2/p})$ in Theorem~\ref{thm:convergence_distribution_estimated_graphon} essentially matches that of Proposition~\ref{convergences of myP}, and that when $\rho_n = o(n^{-2/p})$ the limiting distribution of $\sqrt{n} \rho_n^{-|E(R)|}\bigl(\myP{G_n} - \myP{h}\bigr)$ is not known. 

As stated above, one step required for the proofs of our main theorems is to
show that the bootstrap estimate of variance, for both of our procedures, is
close to $\sigma_R^2$. Proposition \ref{convergence of variance} formalizes
this statement.

\begin{Proposition} \label{convergence of variance} 
	For a motif $R$, under the conditions of Theorem \ref{empirical graphon bootstrap}, the empirical graphon $\hat{h} = \hat{h}_{adj}$ satisfies
	\begin{equation}
	\label{eqn:convergence_of_variance_1}
	\sigma^2_R(\hat{h}) \convprob \sigma^2_R,
	\end{equation}
	and under the conditions of Theorem \ref{thm:convergence_distribution_estimated_graphon}, for any graphon estimate $\hat{h}$ which satisfies~\eqref{asmp:convergence_estimation_1},
	\begin{equation}
	\sigma^2_R(\hat{h}) \convprob \sigma^2_R.
	\end{equation}
\end{Proposition}

Finally, in order to show that Proposition \ref{convergence of variance} holds,
we must have that the $P_S(\hat{h})$ terms in $\sigma^2_R(\hat{h})$ converge at
the appropriate rate to the $P_S(h)$ terms in $\sigma^2_R$. Lemmas
\ref{Convergence of subgraph density} and \ref{lem:convergence_subgraph_density_estimated_graphon} establish this, respectively, for the empirical graphon and for an arbitrary estimator that satisfies~\eqref{asmp:convergence_estimation_1}.

\begin{Lemma} \label{Convergence of subgraph density} For any $k$-node motif
	$S$, if
	$\int_{[0,1]^2}{ w^{2|E(S)|}(u_1,u_2) du_{1:2}} < \infty$, and
	$\rho_n = \omega(n^{-1})$, then for $\hat{h} = \hat{h}_{adj}$
	\begin{equation}
	\left|\frac{\myP[S]{\hat{h}} - \myP[S]{h_n}}{\rho_n^{k-1}}\right| \convprob 0.
	\end{equation}
	If additionally $\rho_n = \omega(n^{-\frac{1}{k}})$, then
	\begin{equation}
	\left| \frac{\myP[S]{\hat{h}} - \myP[S]{h_n}}{\rho_n^{|E(S)|}} \right|  \convprob 0.
	\end{equation}
\end{Lemma}
In order to prove~\eqref{eqn:convergence_of_variance_1} in Proposition~\ref{convergence of variance}, we will need to invoke Lemma~\ref{Convergence of subgraph density} with respect to motifs $S$ in the merged copy set $MC(R,k)$, for $k = p,\ldots,2p - 1$. For such motifs $S$, the number of edges $|E(S)|$ may be as large as $2|E(R)|$, and the number of nodes as large as $2p - 1$. This explains why in Theorem~\ref{empirical graphon bootstrap}, we require that $w \in L^{4|E(R)|}([0,1]^2)$ and $\rho_n = \omega(n^{-1/2p})$. On the other hand, the requirements of Lemma~\ref{lem:convergence_subgraph_density_estimated_graphon} are weaker, and translate to weaker requirements in Theorem~\ref{thm:convergence_distribution_estimated_graphon}.

\begin{Lemma} \label{lem:convergence_subgraph_density_estimated_graphon} 
	For any motif $S$, if $\|w\|_{|E(S)|} < \infty$ and $\hat{h}$ is a graphon estimate which satisfies
	\begin{equation}
	\label{asmp:convergence_subgraph_density_estimated_graphon}
	\rho_n^{-1} \|\hat{h} - h_n\|_{|E(S)|} \convprob 0,
	\end{equation}
	then
	\begin{equation}
	\label{eqn:convergence_subgraph_density_estimated_graphon}
	\left| \frac{\myP[S]{\hat{h}} - \myP[S]{h_n}}{\rho_n^{|E(S)|}} \right|  \convprob 0.
	\end{equation}
\end{Lemma}

\section{Proofs}
\label{sec:proofs}

\subsection{Proof of Theorem \ref{empirical graphon bootstrap}}
We can upper bound (\ref{empirical graphon bootstrap formula}) by
\begin{align*}
& \sup_{x}{ \left|\mathbb{P}\left(\frac{\sqrt{m}}{\bar{\rho}_n^{|E(R)|}} \left(\myP{G_m^*} - \myP{\hat{h}} \right) \le x \middle| G_n\right) - \Phi\left(\frac{x}{\sigma_R(\hat{h})}\right) \right|} \\
& + \sup_{x}{ \left| \Phi\left(\frac{x}{\sigma_R(\hat{h})}\right) - \Phi\left(\frac{x}{\sigma_R({h})}\right) \right|} \\
& + \sup_{x}{\left|\Phi\left(\frac{x}{\sigma_R({h})}\right)  - \Prob{\frac{\sqrt{n}}{\hat{\rho}_n^{|E(R)|}} \left(\myP{G_n} - \myP{h}\right) \le x}\right|}.
\end{align*}

The third term goes to 0 by Proposition \ref{convergences of myP}. The second
term goes to 0 by Proposition \ref{convergence of variance}. All that remains
is to bound the first term. Note that to do so, we cannot simply invoke Proposition~\ref{convergences of myP}, because the empirical graphon $\hat{h}$ is random and changing with $n$. Instead, we bound the first term using a Berry-Esseen inequality for U-statistics, keeping in mind that $\myP{G_n^*}$ is strictly a U-statistic,
conditional on $G_n$, because once the $\epsilon_i^*$ are specified there is no
more randomness in $\myP{G_n^*}$. To ease notation, for $\mathbf{i} \in S_n(p)$ define
\begin{align}
\epsilon_{\mathbf{i}}^* & = (\epsilon_{i_1}^*,\ldots , \epsilon_{i_p}^*)\\
\hat{H}_R(\epsilon_{\mathbf{i}}^*) & := \Expect{\mathbb{I}(G_n^*(\mathbf{i}) = R)\middle| \epsilon^*,G_n}
\end{align}

and note that 

\begin{align}
\myP{G_m^*} & = \frac{1}{N(R) p! {n \choose p}} \sum_{\iset \in S_n(p)} \sum_{R_1 \sim R} \hat{H}_R(\epsilon_{\mathbf{i}}^*)
\end{align}

Therefore, by \citet{Janssen-CLT-for-von-Mises-statistics}, we have that
\begin{equation}
\sup_{x}{\left| \Prob{\frac{\sqrt{m}}{\bar{\rho}_n^{|E(R)|}} \left(\myP{G_m^*} - \myP{\hat{h}} \right) \le x \middle| G_n} - \Phi\left(\frac{x}{\sigma_R(\hat{h})}\right) \right| \le C \frac{\bar{\nu_3}}{\bar{\sigma}_g^{3} m^{\frac{1}{2}}}}
\end{equation}
where
\begin{eqnarray}
\bar{\sigma}_g^{2} & := & \Var{ \Expect{\sum_{R_1 \sim R} \frac{\hat{H}_{R_1}(\epsilon_{1:p}^*)}{p! N(R) \bar{\rho}_n^{|E(R)|}} \middle|   \epsilon_1^*,G_n} \middle| G_n} \\
\bar{\nu_3} &:= & \Expect{ \left| \sum_{R_1 \sim R}{\frac{\hat{H}_{R_1}(\epsilon_{1:p}^*) - \myP{\hat{h}}}{p! N(R) \bar{\rho}_n^{|E(R)|}}} \right|^3 \middle|  G_n }
\end{eqnarray}

We first bound $\bar{\nu}_3$, using the fact that $\hat{H}_R(\epsilon_{1:p}^*)^3 = \hat{H}_R(\epsilon_{1:p}^*)$. 
By Holder's Inequality, we have
\begin{equation}
\Expect{\hat{H}_{R_1}(\epsilon_{1:p}^*)\hat{H}_{R_2}(\epsilon_{1:p}^*)\hat{H}_{R_3}(\epsilon_{1:p}^*)} \le \Expect{\hat{H}_{R}(\epsilon_{1:p}^*)^3}
\end{equation}
By Lemma \ref{Convergence of subgraph density},  we have that
\begin{equation}
\left|\frac{\myP[R]{\hat{h}} - \myP[R]{h_n}}{\rho_n^{|E(R)|}}\right| \convprob 0
\end{equation}
since either $|E(R)| = p - 1$ or $\rho_n = \omega(n^{-\frac{1}{p}})$. Therefore,

\begin{equation}
\myP{\hat{h}} = O_p\left(\myP{h_n}\right) = O_p\left(\rho_n^{|E(R)|}\right)
\end{equation}
where the last statement is implied by the condition $\int_{[0,1]^2}{ w^{4|E(R)|}(u_1,u_2) du_{1:2}} < \infty$.
Putting these together, we can upper bound $\bar{\nu}_3$,
\begin{align}
\bar{\nu_3} & \leq \frac{8}{\bar{\rho}_n^{|3E(R)|}} (\Expect{\hat{H}_R(\epsilon_{1:p}^*)^3 \middle| G_n}  +  \myP{\hat{h}}^3) \\
& = \frac{8}{\bar{\rho}_n^{|3E(R)|}} (\Expect{\hat{H}_R(\epsilon_{1:p}^*) \middle| G_n}  +  \myP{\hat{h}}^3) \\
& = \frac{8}{\bar{\rho}_n^{|3E(R)|}} (\myP{\hat{h}} + \myP{\hat{h}}^3) \\
& = O_p\left(\rho_n^{-2|E(R)|}\right).
\end{align}
Turning to $\bar{\sigma}_g^2$, we have
\begin{align}
\bar{\sigma}_g^2 &= \Var{\Expect{\frac{\sum_{R_1 \sim R} \hat{H}_{R_1}(\epsilon_{1:p}^*)}{p! N(R) \bar{\rho_n}^{|E(R)|}}\middle| \epsilon_1^*, G_n} \middle| G_n} \\
&= \frac{\sum_{S \in MC(R,1)}C_R(S) \left(\myP[S]{\hat{h}} - \myP{\hat{h}}^2 \right) }{\bar{\rho_n^{2|E(R)|}}},
\end{align}
with the second equality following because $$\Expect{\Expect{\hat{H}_{R_1}(\epsilon_{1:p}^*)|\epsilon_1^*}\Expect{\hat{H}_{R_2}(\epsilon_{1:p}^*)|\epsilon_1^*} \middle| G_n} = \myP[S]{\hat{h}}$$ for some $S$ in $MC(R,1)$. If $R$ is acyclic,
$2|E(R)| = 2p - 2 = |V(S)|-1$. Otherwise, by assumption,
$\rho_n = \omega(n^{-\frac{1}{2p}})$ and $|V(S)| = 2p-1$. Either way, by Lemma
\ref{Convergence of subgraph density}, we have that for all $S \in MC(R,1)$,
\begin{equation}
\left|\frac{\myP[S]{\hat{h}} - \myP{\hat{h}}^2 - \myP[S]{h} + \myP{h}^2}{\bar{\rho_n}^{2|E(R)|}}\right| \convprob 0
\end{equation}
Note that $\frac{\myP[S]{h} - \myP{h}^2}{\bar{\rho_n}^{2|E(R)|}} = \myNormP[S]{h} -
(\myNormP{h})^2 = \theta(1)$ by Holder's inequality, unless $w = 1$ almost everywhere. 
So, $\bar{\sigma}_g^2 = \Omega_p(1)$. Finally, if $m = \omega(\rho_n^{-4|E(R)|})$,
\begin{equation}
\bar{\nu_3} \bar{\sigma}_g^{-3} m^{-\frac{1}{2}} = o_p(1)
\end{equation}
which completes the proof.

\subsection{Proof of Theorem~\ref{thm:convergence_distribution_estimated_graphon}}
We begin, similarly to the Proof of Theorem~\ref{empirical graphon bootstrap}, by upper bounding \eqref{eqn:convergence_distribution_estimated_graphon} via the triangle inequality by
\begin{align}
\nonumber & \sup_{x}{\left| \Prob{\frac{\sqrt{n}}{\bar{\rho}_n^{|E(R)|}} \left(\myP{G_n^*} - \myP{\hat{h}} \right) \le x \middle| G_n} - \Phi\left(\frac{x}{\sigma^2_R(\hat{h})}\right) \right|} \\
& + \sup_{x}{\left| \Phi\left(\frac{x}{\sigma^2_R(\hat{h})}\right) - \Phi\left(\frac{x}{\sigma^2_R({h})}\right) \right|} \\
& + \sup_{x}{\left|\Phi\left(\frac{x}{\sigma^2_R({h})}\right)  - \Prob{\frac{\sqrt{n}}{\hat{\rho}_n^{|E(R)|}} \left(\myP{G_n} - \myP{h}\right) \le x}\right|}.
\end{align}

The third term goes to 0 by Proposition \ref{convergences of myP}. The second
term goes to 0 by Proposition \ref{convergence of variance}. To bound the first
term, we split $\myP{G_n^*} - \myP{\hat{h}}$ up into two components, based on
the randomness from resampling latent variables and the randomness from
resampling edges respectively. In other words,
\begin{equation}
\myP{G_n^*} - \myP{\hat{h}} = \myP{G_n^*} - \Expect{\myP{G_n^*}\middle| \epsilon^*,G_n} + \Expect{\myP{G_n^*}\middle| \epsilon^*,G_n} - \myP{\hat{h}}.
\end{equation}

Lemma \ref{lem:berry_esseen_estimated_graphon} establishes that
$\Expect{\myP{G_n^*}\middle| \epsilon^*,G_n} - \myP{\hat{h}}$ obeys, conditional on
$G_n$, a central limit theorem for U-statistics. Lemma \ref{lem:variance_residual_estimated_graphon} establishes that
$\myP{G_n^*} - \Expect{\myP{G_n^*}\middle| \epsilon^*,G_n}$, once appropriately
scaled, has asymptotically neglible contribution to the overall
randomness. Let 
\begin{equation}
\tau_R^2(\hat{h}) := \Var{\frac{\sqrt{n}}{\bar{\rho_n}^{|E(R)|}}\Expect{\myP{G_n^*}\middle| \epsilon^*,G_n} \middle| G_n}
\end{equation}

\begin{Lemma} \label{lem:berry_esseen_estimated_graphon} Let $R$ be a fixed
	motif. Then, if the conditions of Theorem \ref{thm:convergence_distribution_estimated_graphon} hold,
	\begin{equation}
	\sup_{x}{\left| \Prob{ \frac{\sqrt{n}}{\bar{\rho_n}^{|E(R)|}} (\Expect{\myP{G_n^*}\middle| \epsilon^*,G_n} - \myP{\hat{h}}) \le x \middle| G_n} - \Phi\left(\frac{x}{ \tau_R(\hat{h})}\right) \right|} \convprob 0.
	\end{equation}
\end{Lemma}


\begin{proof}
	We use the same notation as in the Proof of Theorem \ref{empirical graphon bootstrap}. We can write
	\begin{align}
	\Expect{\myP{G_n^*}\middle| \epsilon^*,G_n} & = \frac{1}{{n \choose p}} \sum_{\mathbf{i} \in S_n(p)} \frac{\Expect{\mathbb{I}(G_n^*(i) \simeq R)\middle| \epsilon_i^*,G_n}}{p! N(R) \bar{\rho_n}^{|E(R)|}} \\
	& = \frac{1}{{n \choose p}} \sum_{\mathbf{i} \in S_n(p)} \sum_{R_1 \sim R} \frac{\hat{H}_{R_1}(\epsilon_{i}^*)}{p! N(R) \bar{\rho_n}^{|E(R)|}},
	\end{align}
	which shows that $\Expect{\myP{G_n^*}\middle| \epsilon^*,G_n}$ is a U-statistic
	conditional on the graph $G_n$. 
	The Berry-Esseen theorem for U-statistics
	\citep{Janssen-CLT-for-von-Mises-statistics} therefore tells us
	\begin{equation}
	\label{eqn:berry_esseen_u_statistics}
	\begin{split}
	\sup_{x}{ \left| \mathbb{P} \left(\frac{\sqrt{n}}{\bar{\rho_n}^{|E(R)|}} (\Expect{\myP{G_n^*}\middle| \epsilon^*,G_n} - \myP{\hat{h}}) \le x \middle| G_n \right) - \Phi\left(\frac{x}{\tau_R(\hat{h})}\right) \right|}\\
	\le \bar{\nu_3} \bar{\sigma_g}^{-3} n^{-\frac{1}{2}}
	\end{split}
	\end{equation}
	where $\bar{\nu_3}$ and $\bar{\sigma_g}^{-3}$ are defined as in the proof of
	Theorem \ref{empirical graphon bootstrap}. First, we'll upper bound
	$\bar{\nu_3}$.
	\begin{align}
	\bar{\nu_3} & \leq 8 \left(\Expect{\left|\frac{\sum_{R_1 \sim R}\hat{H}_{R_1}(\epsilon_{1:p}^*)}{p! N(R) \bar{\rho_n}^{|E(R)|}}\right|^3 \middle| G_n} + \Expect{\frac{\sum_{R_1 \sim R} \hat{H}_{R_1}(\epsilon_{1:p}^*)}{p! N(R) \bar{\rho_n}^{|E(R)|}} \middle| G_n}^3 \right) \\
	& \leq \frac{16}{\bar{\rho_n}^{3|E(R)|}} \Expect{\hat{H}_R(\epsilon_{1:p}^*)^3 \middle| G_n}
	\end{align}
	where the second line follows from Holder's inequality. Then,
	\begin{align}
	\MoveEqLeft \Expect{\hat{H}_R(\epsilon_{1:p}^*)^3 \middle| G_n} \\
	\nonumber & =  \int_{[0,1]^p} \prod_{(i,j) \in R} \hat{h}(u_i,u_j)^3 \prod_{(i,j) \in K_n\setminus R} (1-\hat{h}(u_i,u_j))^3 du_{1:p} \\
	& \le \int_{[0,1]^p} \prod_{(i,j) \in R} \hat{h}(u_i,u_j)^3 du_{1:p} \\
	& \le \int_{[0,1]^2}{(\hat{h}(u_1,u_2)^{3 |E(R)|}) du_{1:2}}\\
	& =  O\left(\rho_n^{3|E(R)|}\right),
	\end{align}
	where the final equality follows by the boundedness in norm assumed in~\eqref{asmp:convergence_estimation_1}. This implies an upper bound on the skewness, 
	\begin{equation}
	\bar{\nu_3} = \frac{O_p\left(\rho_n^{3|E(R)|}\right)}{\bar{\rho_n}^{3|E(R)|}} = O_p(1).
	\end{equation}
	Turning to $\bar{\sigma}_g^2$, just as in Theorem \ref{empirical graphon
		bootstrap}, we have
	\begin{align}
	\bar{\sigma}_g^2 &= \Var{\Expect{\frac{\sum_{R_1 \sim R} \hat{H}_{R_1}(\epsilon_{1:p}^*)}{p! N(R) \bar{\rho_n}^{|E(R)|}}\middle| \epsilon_1^*, G_n} \middle| G_n} \\
	&= \frac{\sum_{S \in MC(R,2p - 1)}C_R(S) \left(\myP[S]{\hat{h}} - \myP{\hat{h}}^2 \right)}{\bar{\rho}_n^{2|E(R)|}}.
	\end{align}
	Observe that $S \in MC(R,2p - 1)$ satisfies $|E(S)| = 2|E(R)|$. Consequently, the assumptions $\|w\|_{|E(S)|} < \infty$ and $\rho_n\|h_n - \hat{h}\|_{|E(S)|} = o_p(\rho_n)$ of Lemma~\ref{lem:convergence_subgraph_density_estimated_graphon} hold for all $S \in MC(R,2p - 1)$ as well as for $S = R$, and so do the conclusions of Lemma~\ref{lem:convergence_subgraph_density_estimated_graphon}. From here, reasoning exactly the same as in the proof of Theorem 1, except using Lemma~\ref{lem:convergence_subgraph_density_estimated_graphon} rather than Lemma \ref{Convergence of subgraph density}, implies that $\bar{\sigma}_g^2 = \Omega_p(1)$. Combining this with our upper bound on $\bar{\nu_3}$, we conclude that
	\begin{equation}
	\bar{\nu_3} \bar{\sigma}_g^{-3} n^{-\frac{1}{2}} = O_p(1) O_p(1) n^{-\frac{1}{2}} \convprob 0,
	\end{equation}
	which concludes the proof of~Lemma~\ref{lem:berry_esseen_estimated_graphon}.
\end{proof}

\begin{Lemma} \label{lem:variance_residual_estimated_graphon} Let $R$ be a fixed
	motif. Then, if the conditions of Theorem \ref{thm:convergence_distribution_estimated_graphon} hold,
	\begin{equation}
	\label{eqn:variance_residual_estimated_graphon_1}
	\Var{\frac{\sqrt{n}}{p! N(R) \bar{\rho_n}^{|E(R)|}} \left(\myP{G_n^*} - \Expect{\myP{G_n^*}\middle| \epsilon^*,G_n} \right) \middle| G_n} \convprob 0,
	\end{equation}
	and
	\begin{equation}
	\label{eqn:variance_residual_estimated_graphon_2}
	\frac{\sqrt{n}}{p! N(R) \bar{\rho_n}^{|E(R)|}} \left(\myP{G_n^*} - \Expect{\myP{G_n^*}\middle| \epsilon^*,G_n}\right) \convprob 0.
	\end{equation}
\end{Lemma}
\begin{proof}
	We start by rewriting
	\begin{align}
	\MoveEqLeft    \frac{\sqrt{n}}{p! N(R) \bar{\rho_n}^{|E(R)|}} \left(\myP{G_n^*} - \Expect{\myP{G_n^*}\middle| \epsilon^*,G_n} \right)\\
	\nonumber & = \frac{\sqrt{n}}{p! N(R) {n \choose p}\bar{\rho_n}^{|E(R)|}} \sum_{\mathbf{i} \in S_n(p)} \sum_{R_1 \sim R} \{ \mathbb{I}(G_n^*(\mathbf{i}) = R_1) - \hat{H}_{R_1}(\epsilon_{\mathbf{i}}^*) \}.
	\end{align}
	By the definition
	of $\hat{H}_R$,
	\begin{equation}
	\Expect{\mathbb{I}(G_n^*(\mathbf{i}) = R) - \hat{H}_{R}(\epsilon_{\mathbf{i}}^*)\middle| \epsilon^*,G_n} = 0
	\end{equation}
	and so by the law of total variance,
	\begin{align}
	\MoveEqLeft \text{Cov}\left[\mathbb{I}(G_n^*(\iset) = R_1) -
	\hat{H}_{R_1}(\epsilon_{\iset}^*),\mathbb{I}(G_n^*(\jset) = R_2) -
	\hat{H}_{R_2}(\epsilon_{\jset}^*) \middle| G_n\right] \\
	\nonumber & =  \Expect{\text{Cov} \left[\mathbb{I}(G_n^*(\iset) = R_1) - \hat{H}_{R_1}(\epsilon_{\iset}^*),\mathbb{I}(G_n^*(\jset) = R_2) - \hat{H}_{R_2}(\epsilon_{\jset}^*)\middle| \epsilon^*,G_n \right\}}\\
	& =  \Expect{\text{Cov} \left[\mathbb{I}(G_n^*(\iset) = R_1),\mathbb{I}(G_n^*(\jset) = R_2)\middle| \epsilon^*,G_n \right]}
	\end{align}
	Therefore,
	\begin{align}
	\MoveEqLeft \Var{\frac{\sqrt{n}}{p! N(R) \bar{\rho_n}^{|E(R)|}} \left(\myP{G_n^*} -
		\Expect{\myP{G_n^*}\middle| \epsilon^*,G_n} \right) \middle| G_n} \label{expanded expression of variance}\\
	\nonumber    & =
	\frac{n}{(N(R) p! {n \choose p})^2
		\bar{\rho_n}^{2|E(R)|}} \\
	& \quad \cdot\sum_{\mathbf{i}, \mathbf{j} \in S_n(p)} \sum_{R_1,R_2 \sim R}
	\Expect{\text{Cov} \left[\mathbb{I}(G_n^*(\iset) = R_1),\mathbb{I}(G_n^*(\jset) = R_2)\middle| \epsilon^*,G_n \right]}
	\end{align}  
	Let us fix $\iset,\jset$ and denote by $k = |\iset \cap \jset|$ the number of nodes that $\iset$ and $\jset$ have in common, and $k' = |\iset \cup \jset| = 2p - k$. Note that if $k < 2$,
	then $G_n^*(\iset)$ and $G_n^*(\jset)$ share no dyads, and are thus independent
	once we condition on $\epsilon^*$. Otherwise if $k =2,\ldots,p$, we can bound the expected conditional covariance in terms of the moment $P_{W}(\hat{h})$ for some motif $W \in MC(R,k')$\footnote{In particular, the motif $W$ on nodes $1,\ldots,k'$ such that $W(\iset) = R_1$ and $W(\jset) = R_2$} as follows:
	\begin{align}
	\MoveEqLeft \biggl| \mathbb{E}\Bigl[\text{Cov}\bigl[\mathbb{I}(G_n^*(\iset) = R_1),
	\mathbb{I}(G_n^*(\jset) = R_2) \big| \epsilon^*,G_n\bigr] \Big|G_n\Bigr] \biggr| \\
	\nonumber & \le \mathbb{E}\Bigl[\mathbb{I}(G_n^*(\iset) = R_1),
	\mathbb{I}(G_n^*(\jset) = R_2) \Big|G_n\Bigr] \label{pf:variance_residual_estimated_graphon_1}\\
	& = \myP[W]{\hat{h}}\\
	& = O_p\left(\rho_n^{|E(W)|}\right) \label{pf:variance_residual_estimated_graphon_2},
	\end{align}
	where the last equality follows from Lemma \ref{lem:convergence_subgraph_density_estimated_graphon}, which we may invoke because~\eqref{asmp:convergence_estimation_1} implies $\|\hat{h} - h\|_{|E(W)|} = \leq \|\hat{h} - h\|_{2|E(R)|} = o_p(\rho_n)$ for all $k' = p,\ldots,2p - 1$ and $W \in MC(R,k')$.  
	
	From here, we divide our analysis into cases, based on whether (a) $R$ is acyclic and $\rho_n = \omega(n^{-1})$, or (b) $R$ is general and $\rho_n = \omega(n^{-2/p})$. Assuming (a), we have only the lower bound $|E(W)| \geq k' - 1$ for each $W \in MC(R,k')$, because $W$ must be connected. Fortunately, this lower bound is enough, since in this case $2|E(R)| = 2(p - 1)$. Noting that there will be on the order of $n^{k'}$ valid choices for $\iset$ and $\jset$ which yield two subgraphs with $k$ vertices in
	common, it follows from~\eqref{expanded expression of variance},~\eqref{pf:variance_residual_estimated_graphon_2}, and the facts we have just observed that
	\begin{align}
	& \Var{\frac{\sqrt{n}}{p! N(R) \bar{\rho_n}^{|E(R)|}} \left(\myP{G_n^*} -
		\Expect{\myP{G_n^*}\middle| \epsilon^*,G_n} \right) \middle| G_n} \\
	 & \quad = \sum_{k = 2}^{p} O(n^{1-k}) O_p\biggl(\rho_n^{(k' - 1) - 2|E(R)|}\biggr) \\
	 & \quad = \sum_{k = 2}^{p} O(n^{1-k}) O_p\biggl(\rho_n^{-(k - 1)}\biggr) \\
	 & \quad = o_p(1),
	\end{align}
	which implies~\eqref{eqn:variance_residual_estimated_graphon_1}. 
	
	Otherwise we assume (b), that $R$ is general and $\rho_n = \omega(n^{-2/p})$. By the properties of the merged copy set, if $W \in MC(R,k')$ then $|E(W)| \geq 2|E(R)| - {k \choose 2}$. It follows from~\eqref{expanded expression of variance} and~\eqref{pf:variance_residual_estimated_graphon_2} that
	\begin{align}
	& \Var{\frac{\sqrt{n}}{p! N(R) \bar{\rho_n}^{|E(R)|}} \left(\myP{G_n^*} -
		\Expect{\myP{G_n^*}\middle| \epsilon^*,G_n} \right) \middle| G_n} \\
	& \quad = \sum_{k = 2}^{p} O(n^{1-k}) O_p\biggl(\rho_n^{-{k \choose 2}}\biggr) \\
	& \quad = \sum_{k = 2}^{p} o(\rho_n^{(1-k)p/2}) O_p\biggl(\rho_n^{-{k \choose 2}}\biggr) \\
	& \quad = o_p(1),
	\end{align}
	again implying~\eqref{eqn:variance_residual_estimated_graphon_1}. 
	
	Equation~\eqref{eqn:variance_residual_estimated_graphon_2} then follows via an application of (conditional) Chebyshev's Inequality, which we make precise in Lemma~\ref{lem:conditional_chebyshev}, along with the fact that
	\begin{equation*}
	\mathbb{E}\Bigl[P_R(G_n^{\ast}) - \Expect{\myP{G_n^*}\middle| \epsilon^*,G_n} \Big| G_n\Bigr] = 0.
	\end{equation*}
	Thus the proof of Lemma~\ref{lem:variance_residual_estimated_graphon} follows upon proving Lemma~\ref{lem:conditional_chebyshev}.
\end{proof}

\begin{Lemma}
	\label{lem:conditional_chebyshev}
	Let $(X_n),(Y_n)$ be two sequences of random variables. \\
	Suppose $\mathbb{E}[X_n|Y_n] = 0$ and $\mathrm{Var}(X_n|Y_n) = o_p(1)$. Then $X_n = o_p(1)$.
\end{Lemma}
\begin{proof}
	It suffices to show that for any $a,\delta > 0$, $\mathbb{P}(X_n > a) < \delta$ for all $n$ sufficiently large. To begin with, we have that for any $b > 0$:
	\begin{align*}
	\mathbb{P}(X_n > a) & \leq \mathbb{P}\bigl(X_n > a|\mathrm{Var}(X_n|Y_n) \leq b\bigr) + \mathbb{P}(\mathrm{Var}(X_n|Y_n) > b).
	\end{align*} 
	By the law of iterated expectation, the conditional zero-mean property $\mathbb{E}[X_n|Y_n] = 0$ and Chebyshev's inequality,
	\begin{align*}
	\mathbb{P}\bigl(X_n > a|\mathrm{Var}(X_n|Y_n) \leq b\bigr) & = \mathbb{E}\Bigl[\mathbb{P}(X_n > a|Y_n) \Big| \mathrm{Var}(X_n|Y_n) \leq b \Bigr] \\
	& \leq \mathbb{E}\Bigl[\mathrm{Var}(X_n|Y_n)a^{-2} \Big| \mathrm{Var}(X_n|Y_n) \leq b \Bigr] \\
	& \leq \frac{b}{a^2};
	\end{align*}
	note that the second line follows because by assumption $\mathbb{E}[X_n|Y_n] = 0$. Taking $b = \delta a^2/2$, we have by assumption that $\mathbb{P}(\mathrm{Var}(X_n|Y_n) > b) \leq \delta/2$ for all $n$ sufficiently large, and so for all such $n$ we obtain that $\mathbb{P}(X_n > a) \leq \delta$, as desired.
\end{proof}

The proof of Theorem~\ref{thm:convergence_distribution_estimated_graphon} now follows straightforwardly. Putting Lemmas \ref{lem:berry_esseen_estimated_graphon} and \ref{lem:variance_residual_estimated_graphon} together via Slutsky's Theorem yields
\begin{equation}
\sup_{x}{ \left| \mathbb{P} \left(\frac{\sqrt{n}}{\bar{\rho_n}^{|E(R)|}} (\myP{G_n^*} - \myP{\hat{h}}) \le x \middle| G_n \right) - \Phi \left(\frac{x}{\tau_R(\hat{h})} \right) \right|} \convprob 0
\end{equation}
Finally, we have by the definition of conditional expectation that
\begin{equation}
\sigma_R^2(\hat{h}) = \tau_R^2(\hat{h}) + \Var{\frac{\sqrt{n}}{\bar{\rho_n}^{|E(R)|}}(\myP{G_n^*} - \Expect{\myP{G_n^* \middle| G_n}\middle| \epsilon^*,G_n})\middle| G_n}
\end{equation}

By Lemma \ref{lem:variance_residual_estimated_graphon}, we therefore have that
\begin{equation}
\left|\sigma_R^2(\hat{h}) - \tau_R^2(\hat{h})\right| \convprob 0
\end{equation}
Since $\sigma_R^2(\hat{h}) = \theta_p(1)$ and thus $\tau_R^2(\hat{h}) = \theta_p(1)$, this in turn implies
\begin{equation}
\left|\Phi \left(\frac{x}{ \tau_R(\hat{h})} \right) - \Phi \left(\frac{x}{\sigma_R(\hat{h})} \right)\right| \convprob 0
\end{equation}
which completes the proof of Theorem \ref{thm:convergence_distribution_estimated_graphon}.

\subsection{Proof of Lemma \ref{joint sum over motifs}}

This proof will be made slightly easier by introducing a second motif $S$, also
on $p$ nodes.  (We can think of $S$ as being an isomorphic copy of $R$.)  Let
$\mathbf{i} \cup \mathbf{j} = \mathbf{l}$, where (with slight notational
mangling) $\mathbf{l}$ is an ordered k-tuple. Then, define
$R_{\mathbf{i}} = \{(c,d): (a,b) \in R, (i_a,i_b) = (l_c,l_d) \}$ and similarly
$S_{\mathbf{j}} = \{(c,d): (a,b) \in S, (j_a,j_b) = (l_c,l_d) \}$. (Here, we've
done nothing more than taken the two motifs and sent them to the right
vertices as defined by the joint vertex set $\mathbf{l}$.)  Now, let
$W = R_{\mathbf{i}} \cup S_{\mathbf{j}}$.  We would like to relate
$ \mathbb{I}\{G_n(\mathbf{i}) \simeq R\}\mathbb{I}\{G_n(\mathbf{j}) \simeq S\}$
and $\mathbb{I}\{G_n(\mathbf{l}) \simeq W\}$. Unfortunately, they not not quite
equal. After all, if there are some edges {\em between} the vertices only in
$G_n(\mathbf{i})$ and those only in $G_n(\mathbf{j})$ the LHS can still be 1,
but the RHS will clearly be 0.  To fix this, we sum over all these possible
fuller motifs.  Let $ C_{V(R_{{\mathbf{i}}}),V(S_{\mathbf{j}})}$ be the set
of dyads between vertices only in $R_{{\mathbf{i}}}$ and those only in
$S_{{\mathbf{j}}}$. Then,
\begin{equation}
\mathbb{I}\{G_n(\mathbf{i}) = R\}\mathbb{I}\{G_n(\mathbf{j}) = S\} = \sum_{ \substack{W: W = R_{\mathbf{i}} \cup S_{\mathbf{j}} \cup Q \\ Q \subseteq C_{V(R_\mathbf{i}),V(S_\mathbf{j})}}} \mathbb{I} \{G_n(\mathbf{l}) = W \}
\end{equation}
and the relationship between seeing two motifs on different subsets of nodes
and seeing one merged motif on the union of the subsets is
established.\footnote{Notice that we have replaced equality up to isomorphism
	with strict equality. This will simplify the following algebra, and returning
	to the isomorphism relationship can be established with one line at the end.}

These manipulations allow us to write the double sum over $\mathbf{i}$ and
$\mathbf{j}$, with the product of indicators of seeing the motifs $R$ and $S$
on the induced subgraphs $G(\mathbf{i})$ and $G(\mathbf{j})$, as a sum over
$\mathbf{l}$, with the sum of indicators of seeing the motif $W$ on
$G(\mathbf{l})$.
\begin{equation} \label{expansion of double sum}
\begin{aligned}
& \sum_{\mathbf{i},\mathbf{j} \in S_p(n)} \mathbb{I}(G_n(\mathbf{i}) = R) \mathbb{I}(G_n(\mathbf{j}) = S) \\
& \quad = \sum_{k = p}^{2p} \sum_{\mathbf{l} \in S_k(n)} \sum_{\substack{\mathbf{i},\mathbf{j} \in S_p(n): \\ \mathbf{i} \cup \mathbf{j} = \mathbf{l}}} \sum_{\substack{W: W = R_{\mathbf{i}} \cup S_{\mathbf{j}} \cup Q \\ Q \subseteq C_{V(S_\mathbf{i}),V(S_\mathbf{j})}}} \mathbb{I}(G_n(\mathbf{l}) = W).
\end{aligned}
\end{equation}
(\ref{expansion of double sum}) makes clear how we can leverage the assumption
of exchangeability, since
$\Expect{\mathbb{I}(G_n(\mathbf{l}) = W)} = \myP[W]{h}$ remains unchanged
for all choices of $\mathbf{l}$, and so we can simplify (\ref{eqn: joint sum over motifs}) to
\begin{equation}
\sum_{k = p}^{2p} {n \choose k}k! \sum_{\substack{\mathbf{i},\mathbf{j} \in S_p(n): \\ \mathbf{i} \cup \mathbf{j} = 1:k}} \sum_{\substack{W: W = R_\mathbf{i} \cup S_\mathbf{j} \cup Q \\ Q \subseteq C_{V(R_\mathbf{i}),V(S_\mathbf{j})}}} \myP[W]{h}.
\end{equation}

Now, let us specify that $S \simeq R$. Then, for every choice of $\mathbf{i}$,
$\mathbf{j}$ and $Q$, by definition $R_\mathbf{i}, S_\mathbf{j} \simeq R$ and
so $W \in MC(R,k)$. Moreover, for a given $k$ the number of choices of
$\mathbf{i},\mathbf{j}$ and $Q$ are clearly fixed in $n$, and so the number of
times each $W$ in $MC(R,k)$ appears in the sum must also be fixed in $n$.
Finally, to return to isomorphism note that
\begin{equation}
\sum_{R_1, R_2 \simeq R} \mathbb{I}(G_n(\iset) = R_1)\mathbb{I}(G_n(\iset) = R_2) = \mathbb{I}(G_n(\iset) \simeq R)\mathbb{I}(G_n(\iset) \simeq R).
\end{equation}
and of course the number of $S \simeq R$, $N(R)$ is fixed in $n$ as well. Denote the
number of times each $W$ appears as $C_{W}(R)$, where
\begin{equation}
C_{W_0}(R_0) = \sum_{\substack{\iset, \jset \in 1:k, \\ \iset \cup \jset = 1:k}} \sum_{R, S \simeq R_0} \sum_{\substack{W: W = R_\mathbf{\iset} \cup S_\mathbf{\jset} \cup Q \\ Q \subseteq C_{V(R_\mathbf{i}),V(S_\mathbf{j})}}} \mathbb{I}(W = W_0),
\end{equation}
and the expression reduces to exactly the desired form.

\subsection{Proof of Lemma \ref{var_of_p}}

To get the expression for the first statement in the lemma, we expand the
square and use Lemma \ref{joint sum over motifs}.

\begin{align}
\Expect{\myP{G_n}^2} & = \frac{n}{\rho_n^{2|E(R)|} ({n \choose p} p! N(R))^2} \sum_{\mathbf{i},\mathbf{j} \in S_p(n)} \Expect{ \mathbb{I}(G_n(\mathbf{i}) \simeq R) \mathbb{I}(G_n(\mathbf{j}) \simeq R} \\
& = \frac{n}{\rho_n^{2|E(R)|} ({n \choose p} p! N(R))^2} \sum_{k = p}^{2p} {n \choose k} \sum_{W \in MC(R,k)} C_R(W) \myP[W]{h}
\end{align}
Subtracting $P_R(h)^2$ from this, and doing some basic algebraic rearrangement,
yields the desired result.

Now, we turn to the second statement in the lemma. Since the set $MC(R,k)$ is
finite for any given $R$ and $k$, and we are summing over a finite number of
$k$, the problem reduces to showing that
\begin{equation}
\lim{\frac{n {n\choose k} \myP[W]{h}}{(\rho_n^{|E(R)|} {n \choose p})^2}} < \infty
\end{equation}
for all $k$ in $p, \ldots , 2p-1$ and all $W$ in $MC(R,k)$. But then,
\begin{align}
\myP[W]{h} & \le \int_{[0,1]^k}{\prod_{(i,j) \in E(W)} h_n(u_i,u_j) du_{1:k}} \\
& \le \int_{[0,1]^2}{h_n(u_1,u_2)^{|E(W)|} du_{1:2}} \\
& = O\left(\rho_n^{|E(W)|}\right)
\end{align}
and so
\begin{equation}
\underset{n \to \infty}{\lim} \frac{n {n\choose k} \myP[W]{h}}{(\rho_n^{|R|} {n \choose p} p! N(R))^2} = O\left(n^{k+1-2p} \rho_n^{|E(W)| - 2|E(R)|}\right) \convprob 0
\end{equation}
where the last statement follows because either $R$ is acyclic (and thus
$|E(R)| = p-1$) or $\rho_n = O(1/p)$.

\subsection{Proof of Proposition \ref{convergences of myP}}

Both \eqref{convergence in dist of myP} and \eqref{convergence in prob of myP} come
from \citet{Bickel-Chen-Levina-method-of-moments}.

\subsection{Proof of Proposition \ref{convergence of variance}}
By Lemma~\ref{var_of_p}, $\sigma_R^2(h_n) \to \sigma_R^2$; thus to prove Proposition~\ref{convergence of variance} it suffices to show that for either the specific estimator $\hat{h} = \hat{h}_{adj}$, or for an arbitrary estimator $\hat{h}$ that satisfies~\eqref{asmp:convergence_estimation_1},
\begin{equation}
\label{pf:convergence_of_variance_1}
|\sigma_R^2(\hat{h}) - \sigma_R^2(h_n)| \convprob 0.
\end{equation}
Lemma~\ref{var_of_p} also gives an expression for the normalized variance $\sigma_R^2(h)$, when either $h = h_n$ is the true graphon, or $h = \hat{h}$ is a graphon estimate. Taking the difference between these two gives
\begin{align}
& \sigma^2_R(h_n) - \sigma^2_R(\hat{h}) =  \\
& \quad \frac{n}{({n \choose p} p! (N(R))^2} \sum_{k=p}^{2p-1}{ {n \choose k} \sum_{S \in MC(R,k)} C_R(S) \left[\frac{\myP[S]{h_n}}{\rho_n^{2|E(R)|}} - \frac{\myP[S]{\hat{h}}}{\bar{\rho_n}^{2|E(R)|}} \right]}\nonumber \\
\nonumber & \quad - \frac{1}{p!^2 N(R)^2}\left(1 - \frac{{n-p \choose p}}{{n \choose p}} \right)^2 \left[\frac{\myP{h_n}}{\rho_n^{|E(R)|}} - \frac{\myP{\hat{h}}}{\bar{\rho_n}^{|E(R)|}} \right].
\end{align}
The above expression makes clear that in order to establish~\eqref{pf:convergence_of_variance_1}, we want to suitably upper bound $|P_S(\hat{h}) - P_S(h_n)|$ for each motif $S \in MC(R,k), k = p,\ldots,2p - 1$, as well as for $S = R$.

We now collect the relevant upper bounds on~$|P_S(\hat{h}) - P_S(h_n)|$, under various assumptions and for each of $\hat{h} = \hat{h}_{adj}$ or $\hat{h}$ an arbitrary estimator satisfying~\eqref{asmp:convergence_estimation_1}. We then show that these upper bounds imply~\eqref{pf:convergence_of_variance_1}.

\paragraph{Rates of convergence for motif densities.}
We summarize the results of Lemmas~\ref{Convergence of subgraph density} and~\ref{lem:convergence_subgraph_density_estimated_graphon}, insofar as they apply to the proof of Proposition~\ref{convergence of variance}
\begin{itemize}
	\item Suppose $R$ is a general motif, and $S \in MC(R,k)$ for some $k = p,\ldots,2p - 1$. Then by assumption $\|w\|_{|E(S)|} \leq \|w\|_{|2E(R)|} < \infty$, and either (a) $\hat{h} = \hat{h}_{adj}$ and $\rho_n = \omega(n^{-1/2p})$, or (b) $\hat{h} = \hat{h}$ and $\|\hat{h} - h_n\|_{|E(S)|} \leq \|\hat{h} - h_n\|_{2|E(R)|} = o_p(\rho_n)$. Either way, 
	\begin{equation}
	\label{pf:convergence_of_variance_1.25}
	\bigl|P_S(\hat{h}) - P_S(h_n)\bigr| = o_p(\rho_n^{|E(S)|}).
	\end{equation}
	\item Suppose $R$ is an acyclic motif, and $S \in MC(R,k)$ for some $k = p,\ldots,2p - 1$. Then by assumption $\|w\|_{|E(S)|} < \|w\|_{|2E(R)|} < \infty$, and either (a) $\hat{h} = \hat{h}_{adj}$ and $\rho_n = \omega(n^{-1})$, or (b) $\hat{h} = \hat{h}$ and $\|\hat{h} - h_n\|_{|E(S)|} \leq \|\hat{h} - h_n\|_{2|E(R)|} = o_p(\rho_n)$. Either way,\footnote{Of course, Lemma~\ref{lem:convergence_subgraph_density_estimated_graphon} implies something stronger than~\eqref{pf:convergence_of_variance_1.5} in case (b), but we will not need this stronger result.} for any connected $k$-node motif $S$,
	\begin{equation}
	\label{pf:convergence_of_variance_1.5}
	\bigl|P_S(\hat{h}) - P_S(h_n)\bigr| = o_p(\rho_n^{k - 1}).
	\end{equation}
\end{itemize}
Note that the above pair of statements hold because the assumptions of Proposition~\ref{convergence of variance} subsume those of Lemma~\ref{Convergence of subgraph density} and Lemma~\ref{lem:convergence_subgraph_density_estimated_graphon}, respectively. Now we turn to establishing~\eqref{pf:convergence_of_variance_1}; we will show the results for $\hat{h} = \hat{h}_{adj}$ and $\hat{h}$ an arbitrary estimator satisfying~\eqref{asmp:convergence_estimation_1} at the same time, since the proof relies only on the above properties.

\paragraph{Proof of~\eqref{pf:convergence_of_variance_1}.}
To begin, we note that $\bar{\rho}_n = \hat{P}_{S}(\hat{h})$ for $S = K_2$, and so we know that 
\begin{equation}
\label{pf:convergence_of_variance_2}
|\bar{\rho}_n - \rho_n| = o_p(\rho_n).
\end{equation}
Next, we observe that 
\begin{equation}
\label{pf:convergence_of_variance_3}
\left(1 - \frac{{n-p \choose p}}{{n \choose p}} \right)^2\bigl|P_R(\hat{h}) - P_R(h_n)\bigr| = o\Bigl(\bigl|P_R(\hat{h}) - P_R(h_n)\bigr|\Bigr) = o_p\bigl(\rho_n^{|E(R)|}\bigr);
\end{equation}
the last equality follows immediately from~\eqref{pf:convergence_of_variance_1.25} when $R$ is general, and from~\eqref{pf:convergence_of_variance_1.5} in the special case where $R$ is acyclic and $|E(R)| = p - 1$.

It remains to show that for all $S \in MC(R,k)$ and $k = p,\ldots,2p - 1$,
\begin{equation} 
\label{pf:convergence_of_variance_4}
\frac{n {n \choose k} }{{n \choose p}^2 \rho_n^{2|E(R)|}} \left[\myP[S]{h_n} - \myP[S]{\hat{h}} \right] \convprob 0.
\end{equation}
We now separate our analysis into the case where $R$ is a general motif, and the special case of $R$ acyclic. 

\underline{\emph{General motif}.} 
When $R$ is general, we can use \eqref{pf:convergence_of_variance_1.25} to reduce the left hand side of \eqref{pf:convergence_of_variance_4} to
\begin{equation} 
\label{pf:convergence_of_variance_5}
\frac{n {n \choose k} }{{n \choose p}^2 \rho_n^{2|E(R)| - |E(S)|}} \frac{\left[\myP[S]{h_n} - \myP[S]{\hat{h}} \right]}{\rho_n^{|E(S)|}} = O(n^{k + 1 - 2p}) o_p(1) \rho_n^{|E(S)| - 2|E(R)|}.
\end{equation} 
To upper bound $2|E(R)| - |E(S)|$, we leverage the fact
that $S$ is a member of the merged copy set of $R$. The key is to notice
that edges which are lost in $S$ when vertices are merged can only be edges
between two merged vertices. There are $2p-k$ such merged vertices, so there
must be at least $2|E(R)| - {2p-k \choose 2}$ edges in $S$, and
$2|E(R)| - |E(S)| \le \frac{(2p - k)(2p-k-1)}{2} =: \frac{k'(k' - 1)}{2}$ for $k' = 2p - k$. Then plugging back into~\eqref{pf:convergence_of_variance_5} gives
\begin{align}
\frac{n {n \choose k} }{{n \choose p}^2 \rho_n^{2|E(R)|}} \left[\myP[S]{h_n} - \myP[S]{\hat{h}} \right] & =
O(n^{1 - k'}) \rho_n^{-k'(k' - 1)/2} o_p(1) \\
& = O(n^{1 - k'}) O(n^{k' - 1}) o_p(1) = o_p(1),
\end{align}
with the penultimate equality following because $\rho_n = \Omega(n^{-2/p}) = \Omega(n^{-2/k'})$.

\underline{\emph{Acyclic motif}.} 
Otherwise if $R$ is acyclic, then $|E(R)| = p - 1$, and we use~\eqref{pf:convergence_of_variance_1.5} to deduce that 
\begin{align}
\frac{n {n \choose k} }{{n \choose p}^2 \rho_n^{2|E(R)|}} \left[\myP[S]{h_n} - \myP[S]{\hat{h}} \right] & = \frac{n {n \choose k} }{{n \choose p}^2 \rho_n^{2p - k - 1}} \biggl(\frac{\myP[S]{h_n} - \myP[S]{\hat{h}}}{\rho_n^{k - 1}}\biggr)\\
& = \frac{n {n \choose k} }{{n \choose p}^2 \rho_n^{2p - k - 1}} o_p(1) \\ & = O\left((n\rho_n)^{-(2p-k-1)}\right) o_p(1) \\
& = O(1) o_p(1),
\end{align}
with the last line following because $\rho_n = \Omega(n^{-1})$.

Thus we have established~\eqref{pf:convergence_of_variance_4}, which concludes the proof of Proposition~\ref{convergence of variance}.

\subsection{Proof of Lemma \ref{Convergence of subgraph density}}

Throughout the proof of this lemma, $\hat{h} = \hat{h}_{adj}$. First, we will
show that when $\rho_n = \omega(n^{-\frac{1}{k}})$, where $|V(S)| = k$, then
\begin{equation} \label{difference of expected densities}
\left|\frac{\myP[S]{\hat{h}} - \myP[S]{h}}{\rho_n^{|E(S)|}}\right| \convprob 0
\end{equation}

We begin by bounding (\ref{difference of expected densities}) by
\begin{equation} \label{two terms difference of expected densities}
\left|\frac{\myP[S]{\hat{h}} - \myP[S]{h}}{\rho_n^{|E(S)|}}\right| \le \left|\frac{\myP[S]{\hat{h}} - \myP[S]{G_n}}{\rho_n^{|E(S)|}}\right| + \left|\frac{\myP[S]{G_n} - \myP[S]{h}}{\rho_n^{|E(S)|}}\right|
\end{equation}

The second of these is $o_p(1)$ by Proposition \ref{convergences of myP}, along with the assumption that $\int_{[0,1]^2}w(u,v)^{2|E(S)|}$ is finite.
To bound the first term, we will make use of the
following combinatorial identity, which relates the $V$-statistic $P_S(\hat{h})$ to the $U$-statistic $P_S(G_n)$ and can be verified by standard counting arguments:
\begin{equation} \label{combinatorial identity for P-hat}
\left|\frac{\myP[S]{\hat{h}} - \myP[S]{G_n}}{\rho_n^{|E(S)|}}\right| = \left[1-\frac{{n \choose k} }{n^k}\right]\frac{\myP[S]{G_n}}{\rho_n^{|E(S)|}} + O\biggl(\sum_{j=1}^{k-1} n^{j-k} \sum_{W \in M(S,j)} \frac{\myP[W]{G_n}}{\rho_n^{|E(S)|}}\biggr).
\end{equation}
Here $M(S,j)$ is the set of motifs $W$ on $j$ vertices which can be formed by
merging vertices in $S$, and is not the same as the merged copy set $MC(S,j)$. 

$W$ being formed by merging vertices in $S$ restricts how many fewer edges it
may have than $S$. The first merger of two vertices can have merged at most
$k-1$ edges, the second merger can have merged at most $k-2$ edges, and so
forth. As a result, if $W \in M(S,j)$,
\begin{equation}
|E(S)| - |E(W)| \leq \frac{(k+j-1)(k-j)}{2}.
\end{equation}

We will also use the fact that $\frac{\myP[W]{G_n}}{\rho_n^{|E(W)|}} = O_p(1)$,
again by Proposition \ref{convergences of myP} along with the fact that $W$ has fewer edges than $S$. Putting these two together, we have
\begin{align}
O\left(n^{j-k}\right) \frac{\myP[W]{G_n}}{\rho_n^{|E(S)|}} & = O_p\left(\frac{n^{j-k}}{\rho_n^{\frac{(k+j-1)(k-j)}{2}}}\right) \\
& = o_p\left(n^{j-k + \frac{(k+j-1)(k-j)}{2k}}\right) \\
& = o_p\left(n^{-\frac{k-j}{2k} }\right) = o_p(1)
\end{align}
since $j < k$. We must now deal with the leading term, but this is merely
\begin{equation}
\left[1 - \frac{{n \choose k} k!}{n^k}\right]\frac{\myP[S]{G_n}}{\rho_n^{|E(S)|}} = O\left(\frac{1}{n}\right)O_p(1)
\end{equation}
and so we have shown (\ref{difference of expected densities}) in the case where
$\rho_n = \omega(n^{-1/k})$. 

Now we turn to the setting where the only restriction on $\rho_n$ is that
$n\rho_n \overset{n}{\to} \infty$. What we must show is that
\begin{equation} \label{general bound on order of myP}
\left|\frac{\myP[S]{h}  - \myP[S]{\hat{h}}}{\rho_n^{k-1}}\right| \convprob 0.
\end{equation}
Similar to \ref{combinatorial identity for P-hat}), we have
\begin{equation}
\left|\frac{\myP[S]{\hat{h}} - \myP[S]{G_n}}{\rho_n^{|k-1|}}\right| = \left[1-\frac{{n \choose k} }{n^k}\right]\frac{\myP[S]{G_n}}{\rho_n^{|k-1|}} + O\biggl(\sum_{j=1}^{k-1} n^{j-k} \sum_{W \in M(S,j)} \frac{\myP[W]{G_n}}{\rho_n^{|k-1|}}\biggr)
\end{equation}

By Proposition \ref{convergences of myP}, we have that for all $j$ and all
$W \in MC(S,j)$, $\frac{\myP[W]{G_n}}{\rho_n^{|E(W)|}} = O_p(1)$. Since $|E(W)| \ge j - 1$, this implies
$\frac{\myP[W]{G_n}}{\rho_n^{|k-1|}} = O_p(\rho_n^{j-k})$. Therefore, along with the fact $n \rho_n \to \infty$, it follows that
\begin{equation}
O\left(n^{j-k}\right) \sum_{W \in M(S,j)} \frac{\myP[W]{G_n}}{\rho_n^{|k-1|}} = o_p(1).
\end{equation}

Finally,
$\left[1-\frac{{n \choose k} }{n^k}\right]\frac{\myP[S]{G_n}}{\rho_n^{|k-1|}} = o_p(1)$, since $\frac{\myP[S]{G_n}}{\rho_n^{|E(S)|}} = O_p(1)$ and $|E(S)| \ge k - 1$. Thus, we have shown (\ref{general bound
	on order of myP}).

\subsection{Proof of Lemma~\ref{lem:convergence_subgraph_density_estimated_graphon}}
Throughout the proof of this Lemma, $\hat{h}$ is a graphon estimate which satisfies~\eqref{asmp:convergence_subgraph_density_estimated_graphon} but is otherwise arbitrary, and $q = |V(S)|$. We can rewrite $\myP[S]{\hat{h}} - \myP[S]{h_n}$ using their definitions
to yield
\begin{align} \label{definition of myPhat - myP}
\Biggl| & \int_{[0,1]^q} \biggr\{ \prod_{(i,j) \in E(S)} \hat{h}(u_i,u_j) \prod_{(i,j) \not\in E(S)} (1-\hat{h}(u_i,u_j)) \\
& \quad - \prod_{(i,j) \in E(S)} h_n(u_i,u_j) \prod_{(i,j) \not\in E(S)} (1-h_{n}(u_i,u_j)) du_{1:q} \biggr\} \Biggr|.
\end{align}

We now have the difference of products, but want the product of differences. We would also like to ignore the contribution of non-edges. So we use the following Lemma.
\begin{Lemma} \label{lem:bound_of_products} 
	Let $a_1,\ldots,a_\ell, b_1,\ldots,b_\ell \in [0,1]$ for a positive integer $\ell$. Then, for any $k \in 1:\ell$,
	\begin{equation}
	\label{eqn:bound_of_products}
	\begin{aligned}
	& \biggl|\prod_{i = 1}^{k}a_i\prod_{i = k + 1}^{\ell}(1 - a_i) - \prod_{i = 1}^{k}b_i\prod_{i = k + 1}^{\ell}(1 - b_i)\biggr| \\
	& \quad \leq \prod_{i = 1}^{k - 1}a_i \cdot \sum_{j = k + 1}^{\ell}|a_j - b_j| + \sum_{j = 1}^{k}\biggl\{\prod_{i=1}^{j - 1} a_i \cdot |a_j - b_j| \cdot \prod_{i = j + 1}^{k} |b_i| \biggr\}.
	\end{aligned}
	\end{equation}
\end{Lemma}
In~\eqref{eqn:bound_of_products} we have adopted the convention that products which run over empty index sets are one---i.e. $\prod_{i = 2}^{1} a_i = \prod_{i = k + 1}^{k}b_i = 1$--- and sums which run over empty index sets are zero---i.e. $\sum_{j = \ell + 1}^{\ell} |a_j - b_j| = 0$---for notational conciseness. 

\begin{proof} \emph{Lemma~\ref{lem:bound_of_products}}
	To begin, suppose $k = \ell$, so that our goal is to show
	\begin{equation}
	\label{pf:bound_of_products_1}
	\biggl|\prod_{i = 1}^{k}a_i - \prod_{i = 1}^{k}b_i\biggr| \leq \sum_{j = 1}^{k}\biggl\{\prod_{i=1}^{j - 1} a_i \cdot |a_j - b_j| \cdot \prod_{i = j + 1}^{k} |b_i| \biggr\}.
	\end{equation}
	When $k  = \ell = 1$, the claim is obvious. For a general $k = \ell$ it follows by induction; letting $A = \prod_{i = 1}^{k - 1}a_i$ and $B = \prod_{i = 1}^{k - 1}b_i$, we have
	\begin{align*}
	|Aa_k - Bb_k| & \leq A|a_k - b_k| + b_k|A - B| \\
	& \leq A|a_k - b_k| + b_k \sum_{j = 1}^{k - 1}\biggl\{\prod_{i=1}^{j - 1} a_i \cdot |a_j - b_j| \cdot \prod_{i = j + 1}^{k - 1} |b_i| \biggr\} \\
	& = \sum_{j = 1}^{k}\biggl\{\prod_{i=1}^{j - 1} a_i \cdot |a_j - b_j| \cdot \prod_{i = j + 1}^{k} |b_i| \biggr\}.
	\end{align*}
	
	Now if $k < \ell$, then we have
	\begin{align*}
	& \Biggl|\prod_{i = 1}^{k}a_i\prod_{i = k + 1}^{\ell}(1 - a_i) - \prod_{i = 1}^{k}b_i\prod_{i = k + 1}^{\ell}(1 - b_i)\Biggr| \\
	& \quad \leq \prod_{i = 1}^{k}a_i\Biggl|\prod_{i = k + 1}^{\ell}(1 - a_i) - \prod_{i = k + 1}^{\ell}(1 - b_i)\Biggr| + \prod_{i = k + 1}^{\ell}(1 - b_i) \Biggl|\prod_{i = 1}^{k}a_i - \prod_{i = 1}^{k}b_i\Biggr| \\
	& \quad \leq \prod_{i = 1}^{k}a_i \sum_{j = k + 1}^{\ell} |a_j - b_j| + \Biggl|\prod_{i = 1}^{k}a_i - \prod_{i = 1}^{k}b_i\Biggr|,
	\end{align*}
	with the second inequality following because $1 - b_i \in [0,1]$, along with the sum-product inequality which holds for numbers in $[0,1]$. Then the claim of the Lemma follows from~\eqref{pf:bound_of_products_1}.
\end{proof}

For notational convenience, let us adopt an arbitrary ordering $e_1,\ldots,e_{|E(S)|}$ of the edges in $E(S)$, and also write $\hat{\Delta}(u,v) := |\hat{h}(u,v) - h_n(u,v)|$. From~\eqref{definition of myPhat - myP} and Lemma~\ref{lem:bound_of_products}, we can upper bound $|P_S(\hat{h}) - P_S(h_n)|$ by the sum of two terms,
\begin{align*}
& |P_S(\hat{h}) - P_S(h_n)| \\
& \quad \leq \sum_{(i',j') \not\in E(S)} \int_{[0,1]^q} \prod_{i = 1}^{|E(S)| - 1} h_n(u_{e_i(1)}, u_{e_i(2)}) \hat{\Delta}(u_{i'},u_{j'}) \,du_{1:q} \\
& \quad \quad + \sum_{j = 1}^{|E(S)|} \int_{[0,1]^q} \prod_{i = 1}^{j - 1}h_n(u_{e_i(1)}, u_{e_i(2)}) \hat{\Delta}(u_{e_i(1)}, u_{e_i(2)}) \prod_{i = j + 1}^{k}\hat{h}(u_{e_i(1)}, u_{e_i(2)}) \\
& \quad =: T_1 + T_2.
\end{align*}
Thus it remains to show that $T_1,T_2 = o_p(\rho_n^{|E(S)|})$. This is accomplished in a similar manner for each term, by using H\"{o}lder's inequality, the assumption $\|w\|_{|E(S)|} < \infty$, and the rate of convergence assumed in~\eqref{asmp:convergence_subgraph_density_estimated_graphon}. For $T_1$, we have
\begin{align*}
T_1 & = \sum_{(i',j') \not\in E(S)} \int_{[0,1]^q} \prod_{i = 1}^{|E(S)| - 1} h_n(u_{e_i(1)}, u_{e_i(2)}) \hat{\Delta}(u_{i'},u_{j'}) \,du_{1:q} \\
& \leq \sum_{(i',j')} \|h_n\|_{|E(S)|}^{|E(S)| - 1} \|\hat{\Delta}\|_{|E(S)|} \\
& = o_p(\rho_n^{|E(S)|}),
\end{align*}
and for $T_2$,
\begin{align*}
T_2 & = \sum_{j = 1}^{|E(S)|} \int_{[0,1]^q} \prod_{i = 1}^{j - 1}h_n(u_{e_i(1)}, u_{e_i(2)}) \hat{\Delta}(u_{e_i(1)}, u_{e_i(2)}) \prod_{i = j + 1}^{k}\hat{h}(u_{e_i(1)}, u_{e_i(2)})   \\
& \leq \sum_{j = 1}^{|E(S)|} \|h_n\|_{|E(S)|}^{j - 1} \|\hat{\Delta}\|_{|E(S)|} \|\hat{h}\|_{|E(S)| - j} \\
& = o_p(\rho_n^{|E(S)|}).
\end{align*}
This concludes the proof of Lemma~\ref{lem:convergence_subgraph_density_estimated_graphon}.

\subsection*{Acknowledgments}
	We are grateful to the participants in the CMU Networkshop for valuable
	suggestions on the content and presentation of our results; to Prof.\ Paul
	Janssen for directing us to \citet{Janssen-CLT-for-von-Mises-statistics}; and
	to Profs.\ Carl T. Bergstrom, Peter J. Bickel, Dean Eckles, Jennifer Neville,
	Art B. Owen and Alessandro Rinaldo for valuable discussions, over the years,
	about network bootstraps. Our work was supported by NSF grant DMS1418124.


\end{document}